\documentclass[fleqn,10pt]{sn-jnl}

\usepackage{amsmath,amssymb,amsfonts,amsthm}
\usepackage{mathtools}   

\usepackage{algorithm}
\usepackage{algorithmic}
\usepackage[numbers]{natbib}

\usepackage{url}
\usepackage{hyperref}
\usepackage{bbm}

\theoremstyle{plain}
\newtheorem{theorem}{Theorem}[section]
\newtheorem{lemma}[theorem]{Lemma}
\newtheorem{proposition}[theorem]{Proposition}
\newtheorem{corollary}[theorem]{Corollary}
\newtheorem{assumption}[theorem]{Assumption}

\theoremstyle{definition}
\newtheorem{definition}[theorem]{Definition}
\newtheorem{remark}[theorem]{Remark}

\title{Sparse Probabilistic Coalition Structure Generation: Bayesian Greedy Pursuit and $\ell_1$ Relaxations}

\author[1,*]{Angshul Majumdar}
\affil[1]{IIIT Delhi, India}
\affil[*]{angshul@iiitd.ac.in}


\begin{document}
\maketitle

\begin{abstract}
WWe study coalition structure generation (CSG) when coalition values are not given but must be learned from episodic observations. We model each episode as a sparse linear regression problem, where the realised payoff \(Y_t\) is a noisy linear combination of a small number of coalition contributions. This yields a probabilistic CSG framework in which the planner first estimates a sparse value function from \(T\) episodes, then runs a CSG solver on the inferred coalition set.
We analyse two estimation schemes. The first, Bayesian Greedy Coalition Pursuit (BGCP), is a greedy procedure that mimics orthogonal matching pursuit. Under a coherence condition and a minimum signal assumption, BGCP recovers the true set of profitable coalitions with high probability once \(T \gtrsim K \log m\), and hence yields welfare-optimal structures. The second scheme uses an \(\ell_1\)-penalised estimator; under a restricted eigenvalue condition, we derive \(\ell_1\) and prediction error bounds and translate them into welfare gap guarantees. We compare both methods to probabilistic baselines and identify regimes where sparse probabilistic CSG is superior, as well as dense regimes where classical least-squares approaches are competitive.
\end{abstract}

\section{Introduction}
\label{sec:intro}

Coalition structure generation (CSG) is a central optimisation problem
in cooperative game theory and multi-agent systems: given a set of
agents $N = \{1,\dots,n\}$ and a characteristic function
$v : 2^{N} \to \mathbb{R}$, the goal is to partition $N$ into
disjoint coalitions so as to maximise the total social welfare
$\sum_{C \in \mathcal{P}} v(C)$ over all coalition structures
$\mathcal{P}$.
The search space has size given by the Bell numbers, which grow
super-exponentially with $n$, and even in the simplest transferable
utility settings, CSG is NP-hard and demands sophisticated algorithmic
machinery~\citep{Sandholm1999Coalition,Rahwan2015CSGSurvey}.
Over the last two decades, an extensive line of work has developed
dynamic programming, branch-and-bound, and anytime algorithms for
exact or approximate CSG, with a strong emphasis on worst-case
complexity and pruning of the coalition structure search
space~\citep{rahwan2007nearoptimal,rahwan2009anytime,michalak2010distributed,banerjee2010csgexternalities,rahwan2012externalities,Rahwan2015CSGSurvey}.
More recent work in Algorithmica and related venues has further
explored structural and complexity aspects of coalition
structures~\citep{brandt2020popularcoalitions,aziz2024paretocoalition}, including
preferences over partitions and refined efficiency concepts.

In parallel, the multi-agent systems and JAAMAS communities have
developed a rich literature on \emph{coalition formation under
uncertainty}, where coalition values are not known a priori but must be
learned or inferred from data.
One influential line of work adopts a Bayesian reinforcement learning
viewpoint: agents repeatedly form coalitions, observe noisy payoffs,
and update beliefs about the types or capabilities of their potential
partners, leading to new concepts such as the Bayesian core and
sequentially optimal coalition formation policies~\citep{chalkiadakis2004bayesianCF,chalkiadakis2012sequentialCF}.
These models explicitly acknowledge that, in many applications,
coalition values are only revealed through episodic interaction, and
they provide principled ways to balance exploration and exploitation
in repeated coalition formation.
More generally, surveys and case studies in cooperative and
multi-objective MAS
highlight coalition formation as a key mechanism for distributed
decision making and resource allocation, with applications ranging from
energy systems to communication networks and robotics~\citep{sarkar2022coalitionSurvey,radulescu2020multiobjMAS}.

The work on \emph{probabilistic} CSG (PCSG) provides a complementary
formalism in which the source of uncertainty lies in stochastic agent
attendance rather than in noisy payoffs.
In the PCSG model of \citep{schwind2018pcsg}, each agent is active in
a given instance with a specified probability, and the goal is to find
a coalition structure maximising \emph{expected} social welfare under
this probabilistic attendance model.
Subsequent work has designed approximation algorithms and studied the
complexity of PCSG in more detail, including two approximation
schemes for PCSG and related variants in the JAAMAS
literature~\citep{matsumura2020approxPCSG,schwind2021pcsgComputation}.
These contributions firmly place PCSG alongside classical CSG as a
fundamental combinatorial optimisation problem, but they focus
primarily on the algorithmics of computing good coalition structures
within the probabilistic attendance model itself.

In many real-world settings, however, the uncertainty is neither purely
combinatorial (random attendance) nor purely epistemic in the sense of
unknown but fixed coalition values.
Instead, system designers often observe \emph{episodic} data of the
form
\[
  (X_t, Y_t), \qquad t = 1,\dots,T,
\]
where $X_t \in \{0,1\}^{m}$ encodes which of $m$ candidate coalitions
(or coalition features) are active in episode $t$, and
$Y_t \in \mathbb{R}$ is a noisy scalar reward generated by the
underlying characteristic function.
Examples include online advertising and sponsored search, where
different subsets of advertisers are shown in each round; logistics and
transport, where subsets of carriers collaborate on shared loads; and
robotic or sensing teams, where only a few coalition templates are
actually instantiated in each episode.
In these domains, one rarely observes full values $v(C)$ for all
$C \subseteq N$; instead, one observes a sequence of noisy aggregate
payoffs and must both \emph{learn} about the underlying cooperative
structure and \emph{exploit} it to compute high-welfare coalition
structures.
This motivates a probabilistic CSG framework that is explicitly
\emph{data-driven} and statistical in nature.

From a modelling perspective, one natural way to exploit structure in
episodic observations is to posit that only a small number of coalitions
have non-negligible marginal contributions to welfare.
Formally, if we index a collection of $m$ candidate coalitions by
$j \in [m]$ and write $\theta^\star_j$ for the (unknown) contribution
of coalition $j$, then the assumption that only $K \ll m$ of these are
truly relevant corresponds to a sparse parameter vector
$\theta^\star \in \mathbb{R}^{m}$.
The observed episodic payoffs can then be modelled as
\[
  Y_t = X_t^\top \theta^\star + \varepsilon_t,
\]
with sub-Gaussian noise $\varepsilon_t$, exactly in the spirit of
high-dimensional sparse linear models familiar from compressive sensing
and statistical learning.
In this view, probabilistic CSG becomes a \emph{two-level} problem:
(1) \emph{statistically}, we must recover the support and magnitude of
$\theta^\star$ from $T$ episodes, ideally using $T$ that scales only
logarithmically with $m$; and
(2) \emph{combinatorially}, we must then run a (possibly expensive) CSG
solver on the reduced coalition universe suggested by our estimator.

The existing PCSG and coalition-formation literature provides only
partial guidance for such episodic, high-dimensional problems.
Classical CSG algorithms assume that $v(C)$ is either known exactly or
can be queried at unit cost, and they focus on pruning the space of
coalition structures rather than learning $v$ from data~\citep{Rahwan2015CSGSurvey,rahwan2007nearoptimal,rahwan2012externalities,michalak2010distributed,banerjee2010csgexternalities,rahwan2009anytime,rahwan2010hybrid}.
PCSG algorithms, in turn, optimise expected welfare under known
attendance probabilities and do not address noisy observations of
episode-level payoffs~\citep{schwind2018pcsg,matsumura2020approxPCSG,schwind2021pcsgComputation}.
Bayesian reinforcement learning models for coalition formation
explicitly treat uncertainty and repetition, but they typically work
with relatively small coalition spaces and focus on learning stable
coalitions or equilibria rather than on high-dimensional sparse value
functions~\citep{chalkiadakis2004bayesianCF,chalkiadakis2012sequentialCF}.
Finally, algorithmic work in Algorithmica and related journals on
coalition structures in hedonic and voting games emphasises stability,
fairness, and Pareto efficiency
\citep{brandt2020popularcoalitions,aziz2024paretocoalition} rather than
sample-efficient welfare maximisation from episodic data.

This paper takes a different, explicitly \emph{probabilistic} and
\emph{high-dimensional} view of coalition structure generation.
We assume that coalition-level contributions are encoded in a sparse
linear model, and that a planner observes a finite number of noisy
episodes before committing to a coalition structure.
We then study two families of probabilistic CSG schemes:
(i) a Bayesian greedy coalition pursuit (BGCP) algorithm that mimics
Bayesian versions of orthogonal matching pursuit in sparse regression,
and (ii) an $\ell_{1}$-penalised episodic estimation pipeline inspired
by high-dimensional Lasso theory.
For both, we provide non-asymptotic welfare guarantees that are sharp in
the sparse regime and we explicitly compare their statistical and
computational trade-offs with natural episodic plug-in schemes and dense
least-squares baselines.
In doing so, we bridge ideas from high-dimensional statistics and sparse
approximation with the algorithmic CSG literature, and we position
probabilistic CSG as a statistical-computational problem in its own
right, distinct from but complementary to attendance-based PCSG and
classical cooperative game-theoretic formulations.

\section{Probabilistic Model and Sparse Formulation}
\label{sec:model}

In this section we make the probabilistic aspect of coalition structure generation explicit.
We first recall the classical (deterministic) formulation, then introduce a stochastic model
for coalition values and availability, and finally show how this leads to a sparse
(\(\ell_0\)) optimisation view that will be exploited by our algorithms.

\subsection{Classical Deterministic CSG in Brief}

Let \(N = \{1,\dots,n\}\) be a finite set of agents.
A \emph{coalition} is any non-empty subset \(S \subseteq N\), and a
\emph{coalition structure} is a partition
\(\mathcal{P} = \{S_1,\dots,S_k\}\) of \(N\) into disjoint coalitions.
In the classical setting, one is given a deterministic characteristic function
\[
  v : 2^N \setminus \{\emptyset\} \to \mathbb{R},
\]
and the goal is to find a coalition structure \(\mathcal{P}\) that maximises
the total value
\[
  W(\mathcal{P}) = \sum_{S \in \mathcal{P}} v(S).
\]
This \emph{coalition structure generation} (CSG) problem is well known to be
NP-hard in general and closely related to set partitioning and set packing;
see Sandholm et al.\ and Rahwan et al.\ for worst-case guarantees, algorithmic
techniques, and surveys of the CSG literature~\cite{Sandholm1999Coalition,Rahwan2015CSGSurvey}.
In this deterministic model, both the feasible partitions and the coalition
values \(v(S)\) are fixed and assumed to be known at optimisation time.

\subsection{A Probabilistic View: Random Values and Random Participation}

In many applications, coalition values are uncertain or only observed
through noisy samples.
Examples include repeated markets, task-allocation scenarios, and
learning-based environments, where the value that a coalition generates
depends on random external conditions or unobserved states.
To capture this, we work on a probability space \((\Omega,\mathcal{F},\mathbb{P})\)
and model coalition performance and feasibility as random objects.

\paragraph{Random coalition performance.}
For each non-empty coalition \(S \subseteq N\), we introduce a real-valued
random variable
\[
  V(S) : \Omega \to \mathbb{R},
\]
representing the realised payoff of coalition \(S\) in a single
\emph{episode} (e.g., one auction, one task instance, one time period).
We write
\[
  \mu(S) \coloneqq \mathbb{E}\bigl[V(S)\bigr]
\]
for the latent expected value of coalition \(S\).
The classical CSG model is the special case in which \(V(S) \equiv v(S)\)
almost surely, so that \(\mu(S) = v(S)\) and there is no randomness.

\paragraph{Random availability / participation.}
In many settings, not all coalitions are feasible in every episode (for example,
due to resource constraints, agents dropping out, or context-dependent feasibility).
We therefore introduce binary random variables
\[
  A(S) \in \{0,1\}, \qquad S \subseteq N,\; S \neq \emptyset,
\]
where \(A(S) = 1\) indicates that coalition \(S\) is available or active
in the current episode.
The collection \(A = \bigl(A(S)\bigr)_S\) is a random subset of potential
coalitions; its distribution models uncertainty about feasibility or
participation.

\paragraph{Observed data.}
We assume that we observe \(T\) independent and identically distributed episodes
\[
  \{(A_t, V_t)\}_{t=1}^T,
\]
where \(A_t(S)\) is the availability indicator and \(V_t(S)\) the realised payoff
of coalition \(S\) in episode \(t\).
In practice, we may observe only a subset of coalitions per episode
(e.g., bids actually submitted in an auction); this is naturally modelled by
taking \(A_t(S) = 0\) for coalitions that are infeasible or unobserved in
episode \(t\) (cf.\ combinatorial auctions and winner determination,
e.g.~\cite{Rothkopf1998Manageable,Sandholm2000WinnerDetermination}).

Given this probabilistic model, \emph{probabilistic coalition structure
generation} is the problem of using the observed episodes
\(\{(A_t,V_t)\}_{t=1}^T\) to select a coalition structure \(\mathcal{P}\)
with high expected welfare
\[
  W(\mathcal{P}) \coloneqq \sum_{S \in \mathcal{P}} \mu(S)
\]
under the joint law of \((A,V)\).
The key difference from classical CSG is that the expectations \(\mu(S)\)
are not directly given: they must be inferred statistically from data,
and the space of potentially profitable coalitions may be extremely large.

\subsection{Parametric Sparse Valuation Model}

To obtain a connection with sparse optimisation, we introduce a parametric
representation of coalition values, inspired by combinatorial auctions and
sparse linear models~\cite{Rothkopf1998Manageable,Sandholm2000WinnerDetermination,Donoho2006CS,CandesTao2005Decoding}.

Let \(\mathcal{C} = \{C_1,\dots,C_m\}\) be an index set of candidate coalitions
(or \emph{atoms}) in the coalition space.
For example, \(\mathcal{C}\) might consist of all coalitions up to some fixed
size \(k_{\max}\), or it might be a hand-crafted library of ``interesting''
coalitions obtained from domain knowledge or a preprocessing step.

We assume that, for each episode \(t\), we observe an aggregate welfare
\[
  Y_t = \sum_{j=1}^m X_t(C_j)\,\theta_j + \varepsilon_t,
\]
where
\begin{itemize}
  \item \(X_t(C_j) \in \mathbb{R}\) indicates the (possibly binary) activation
        level of coalition \(C_j\) in episode \(t\); in many applications
        we have \(X_t(C_j) \in \{0,1\}\), but more general real-valued
        designs are allowed and will be useful in the analysis;
  \item \(\theta_j \in \mathbb{R}\) is an unknown parameter representing the
        latent contribution of coalition \(C_j\) to welfare;
  \item \(\varepsilon_t\) is a noise term capturing unmodelled randomness
        (e.g., environmental fluctuations or approximation error).
\end{itemize}
Collecting \(T\) episodes, we obtain the linear model
\[
  Y = X \theta + \varepsilon,
\]
where \(Y \in \mathbb{R}^T\), \(X \in \mathbb{R}^{T\times m}\) (often with
binary entries \(X_{tj} \in \{0,1\}\) in our applications), and
\(\theta \in \mathbb{R}^m\).
For suitable choices of the design matrix \(X\), this subsumes both
classical set-packing / winner-determination formulations (each row encodes
a feasible allocation of coalitions to agents, as in combinatorial auctions
\cite{Rothkopf1998Manageable,Sandholm2000WinnerDetermination}) and
stochastic or sample-based settings where each row corresponds to an observed
episode or scenario.

We regard \(\theta\) as a parameter vector governing coalition contributions,
with an unknown ``true'' value \(\theta^\star \in \mathbb{R}^m\) and associated
support
\[
  S^\star \coloneqq \mathrm{supp}(\theta^\star)
  = \{ j \in \{1,\dots,m\} : \theta^\star_j \neq 0 \}.
\]
The sparse recovery problem is to infer \(S^\star\) (and the corresponding
coefficients) from the joint law of \((X,Y)\).
The probabilistic element is now explicit:
\(X\), \(Y\) and \(\varepsilon\) are random, and \(\theta^\star\) is an
unknown parameter to be inferred.

\subsection{Sparse Prior on Profitable Coalitions}

We are interested in situations where only a small number of coalitions
are truly profitable at scale.
This is both a modelling assumption (few structured patterns of cooperation
matter) and a computational necessity (we cannot afford to reason about an
exponential number of active coalitions).
We therefore posit that \(\theta^\star\) is sparse:
\begin{itemize}
  \item the support \(\mathrm{supp}(\theta^\star) = S^\star\) satisfies
        \(|S^\star| \leq K \ll m\).
\end{itemize}

From a Bayesian perspective, a natural way to encode this is via a
Bernoulli--Gaussian (spike-and-slab) prior, as in Bayesian pursuit
algorithms~\cite{HerzetDremeau2014BayesianPursuit}.
For each \(j\),
\[
  Z_j \sim \mathrm{Bernoulli}(p_j), \qquad
  \theta_j \mid Z_j = 1 \sim \mathcal{N}(0,\tau^2), \qquad
  \theta_j \mid Z_j = 0 \equiv 0,
\]
independently across \(j\).
Here the latent indicators \(Z_j\) select which coalitions are globally
``active'', while the Gaussian slab controls the magnitude of their
contributions.

Under this prior, the maximum a posteriori (MAP) estimate of \((Z,\theta)\)
can be shown, in the small-noise limit, to reduce to an \(\ell_0\)-regularised
least-squares problem of the form
\[
  \min_{\theta \in \mathbb{R}^m} \;
    \frac{1}{2}\,\|Y - X\theta\|_2^2 + \lambda \|\theta\|_0,
\]
for a suitable choice of \(\lambda > 0\) determined by the Bernoulli
parameters \(\{p_j\}\) and the noise variance.
This is precisely the standard sparse regression objective studied in the
compressed sensing and sparse approximation literature
(e.g.~\cite{Donoho2006CS,CandesTao2005Decoding}).

Thus, probabilistic coalition structure generation reduces, at the MAP level,
to identifying a sparse set of globally profitable coalitions from noisy
samples, under a Bayesian model that yields an \(\ell_0\)-penalised
optimisation problem.

\subsection{From Sparse Parameters to Coalition Structures}

Once we have inferred an estimate \(\hat{\theta}\), the corresponding
set of candidate coalitions is
\[
  \widehat{\mathcal{C}} \coloneqq
    \{ C_j \in \mathcal{C} : \hat{\theta}_j \neq 0 \}.
\]
On this reduced set \(\widehat{\mathcal{C}}\), we can now solve a
combinatorial coalition structure generation problem
\[
  \max_{\mathcal{P}} \; \sum_{S \in \mathcal{P}} \hat{\mu}(S)
  \quad \text{subject to} \quad
  \mathcal{P} \text{ is a partition of } N,\;
  S \in \widehat{\mathcal{C}} \text{ for all } S \in \mathcal{P},
\]
where \(\hat{\mu}(S)\) is an estimate of \(\mu(S)\), often taken to be
\(\hat{\theta}_j\) when \(S = C_j\).
This combinatorial problem can be attacked with standard dynamic programming,
branch-and-bound, or CSG-specific algorithms, but now on a sparsified
coalition space~\cite{Sandholm1999Coalition,Rahwan2015CSGSurvey}.

In this sense, probabilistic coalition structure generation has two coupled
components:
\begin{enumerate}
  \item a \emph{statistical phase} (sparse learning), in which we infer
        the set \(\widehat{\mathcal{C}}\) of promising coalitions by solving
        a sparse estimation problem under the probabilistic model for episodes;
  \item a \emph{combinatorial phase} (structure generation), in which we
        compute a high-welfare coalition structure using only coalitions
        in \(\widehat{\mathcal{C}}\).
\end{enumerate}
In later sections, we will provide Bayesian greedy (OMP-style) procedures
and \(\ell_1\)-regularised convex relaxations for the statistical phase,
together with probabilistic recovery guarantees under assumptions on
the random design matrix \(X\) that are analogous in spirit to restricted
isometry or incoherence conditions in compressed sensing
(e.g.~\cite{Donoho2006CS,CandesTao2005Decoding,HerzetDremeau2014BayesianPursuit}).
These results yield high-probability statements of the form that, with
probability at least \(1-\delta\) over the sampled episodes, our algorithms
recover the true set of globally profitable coalitions (or an
\(\varepsilon\)-optimal superset) within a given computational budget.

\section{Bayesian Greedy Pursuit and High-Probability Guarantees}
\label{sec:greedy}

In this section we instantiate the probabilistic sparse model of
Section~\ref{sec:model} with a concrete greedy algorithm for estimating
the sparse parameter vector \(\theta^\star\), and hence the set of
globally profitable coalitions.
The algorithm is an orthogonal matching pursuit (OMP)--type procedure
\cite{Pati1993OMP,TroppGilbert2007OMP,DavenportWakin2010OMP,CaiWang2011OMPNoise},
which we interpret as a greedy approximation to the MAP estimator under
the Bernoulli--Gaussian prior of Section~\ref{sec:model}
(cf.~\cite{HerzetDremeau2014BayesianPursuit}).
We then prove a high-probability support recovery result closely
tailored to our probabilistic coalition structure generation setting.

\subsection{Data Model, Support and Coherence}
\label{subsec:greedy-model}

We retain the linear model introduced in
Section~\ref{sec:model}.
For \(T\) episodes we observe
\begin{equation}
  Y = X \theta^\star + \varepsilon,
  \qquad
  Y \in \mathbb{R}^T,\;
  X \in \mathbb{R}^{T \times m},\;
  \theta^\star \in \mathbb{R}^m,
  \label{eq:data-model}
\end{equation}
where each column \(X_j\) of \(X\) corresponds to the activation pattern
of coalition \(C_j \in \mathcal{C}\) across episodes
(Section~\ref{sec:model}), and \(\varepsilon \in \mathbb{R}^T\) denotes
noise.

The \emph{true support} of the parameter is
\begin{equation}
  S^\star
  \;\coloneqq\;
  \mathrm{supp}(\theta^\star)
  \;=\;
  \{ j \in \{1,\dots,m\} : \theta^\star_j \neq 0 \},
  \qquad |S^\star| = K,
  \label{eq:true-support}
\end{equation}
so that \(K\) is the number of globally profitable coalitions.
We adopt the column normalisation
\begin{equation}
  \|X_j\|_2^2 = T
  \quad\text{for all } j=1,\dots,m,
  \label{eq:column-normalisation}
\end{equation}
and define the mutual coherence of \(X\) by
\begin{equation}
  \mu(X) \;\coloneqq\;
  \max_{\substack{i,j \in \{1,\dots,m\} \\ i \neq j}}
  \frac{|X_i^\top X_j|}
       {\|X_i\|_2 \,\|X_j\|_2}.
  \label{eq:coherence}
\end{equation}
Intuitively, \(\mu(X)\) measures the worst-case correlation between
activation patterns of distinct coalitions.

\subsection{MAP Objective and Bayesian Greedy Coalition Pursuit}
\label{subsec:bgcp-def}

Under the Bernoulli--Gaussian prior of Section~\ref{sec:model}, with
noise variance \(\sigma^2\), the negative log-posterior (up to
\(\theta\)-independent constants) can be written as
\begin{equation}
  \mathcal{L}(\theta)
  \;\propto\;
  \frac{1}{2\sigma^2}\,\|Y - X\theta\|_2^2
  + \lambda \|\theta\|_0,
  \qquad
  \lambda > 0,
  \label{eq:map-objective}
\end{equation}
where \(\|\theta\|_0 = |\mathrm{supp}(\theta)|\).
The MAP estimator thus solves an \(\ell_0\)-penalised least-squares
problem of the form
\begin{equation}
  \min_{\theta \in \mathbb{R}^m}
  \left\{
    \frac{1}{2}\,\|Y - X\theta\|_2^2 + \lambda_0 \|\theta\|_0
  \right\},
  \label{eq:l0-objective}
\end{equation}
for a suitable \(\lambda_0>0\) determined by the prior
\cite{HerzetDremeau2014BayesianPursuit}.
Directly solving~\eqref{eq:l0-objective} is combinatorial and
infeasible at scale.

We therefore consider a greedy approximation: at each step we add one
coordinate \(j\) whose column \(X_j\) is most correlated with the
current residual, and then recompute the least-squares fit on the
selected support, as in classical orthogonal matching pursuit
\cite{Pati1993OMP,TroppGilbert2007OMP,DavenportWakin2010OMP}.

\begin{definition}[Bayesian Greedy Coalition Pursuit (BGCP)]
\label{def:bgcp}
Given \(X\), \(Y\), sparsity budget \(K_{\max}\) and residual threshold
\(\eta \ge 0\), the BGCP algorithm proceeds as follows.
\begin{enumerate}
  \item \emph{Initialisation:}
    Set \(k \gets 0\), \(S^0 \gets \emptyset\), \(r^0 \gets Y\).

  \item \emph{Iteration:}
    While \(k < K_{\max}\) and \(\|r^k\|_2 > \eta\):
    \begin{enumerate}
      \item (\emph{Selection}) For each \(j\in\{1,\dots,m\}\) define
        \[
          c_j^k \;\coloneqq\; \frac{1}{T}\, X_j^\top r^k.
        \]
        Choose
        \[
          j_k \in \arg\max_{j \in \{1,\dots,m\}\setminus S^k}
          |c_j^k|.
        \]
        Update the support
        \[
          S^{k+1} \;\gets\; S^k \cup \{j_k\}.
        \]
      \item (\emph{Projection}) Form \(X_{S^{k+1}}\), the submatrix of
        \(X\) with columns indexed by \(S^{k+1}\), and compute the
        least-squares estimate
        \begin{equation}
          \hat{\theta}^{k+1}_{S^{k+1}}
          \;\coloneqq\;
          \arg\min_{u \in \mathbb{R}^{|S^{k+1}|}}
          \|Y - X_{S^{k+1}} u\|_2^2,
          \qquad
          \hat{\theta}^{k+1}_j = 0
          \;\text{ for } j \notin S^{k+1}.
          \label{eq:ls-update}
        \end{equation}
        Update the residual
        \[
          r^{k+1} \;\gets\; Y - X \hat{\theta}^{k+1}.
        \]
      \item (\emph{Increment}) Set \(k \gets k+1\).
    \end{enumerate}
  \item \emph{Output:}
    Final support \(S^{\mathrm{BGCP}} \gets S^k\) and coefficient
    estimate \(\hat{\theta}^{\mathrm{BGCP}} \gets \hat{\theta}^k\).
\end{enumerate}
We write
\(
  \widehat{\mathcal{C}}^{\mathrm{BGCP}}
  \coloneqq \{C_j \in \mathcal{C} : j \in S^{\mathrm{BGCP}}\}
\)
for the set of coalitions selected by BGCP.
\end{definition}

When \(K_{\max}=K\) and \(\eta = 0\), BGCP performs exactly \(K\)
iterations and produces a \(K\)-sparse estimate.
From the MAP perspective, it may be interpreted as a greedy coordinate
search for a low-cost solution of~\eqref{eq:l0-objective}, similarly to
Bayesian pursuit algorithms in sparse signal recovery
\cite{HerzetDremeau2014BayesianPursuit}.

\subsection{Assumptions and Noise Concentration}
\label{subsec:greedy-assumptions}

We now specify the conditions under which BGCP recovers the true
support \(S^\star\) with high probability.
These are standard in the analysis of OMP
\cite{TroppGilbert2007OMP,DavenportWakin2010OMP,CaiWang2011OMPNoise},
adapted to our notation.

\begin{assumption}[Design, noise and minimum signal]
\label{assump:bgcp}
Consider the model~\eqref{eq:data-model} with support \(S^\star\) of
size \(K\).
Assume:
\begin{enumerate}
  \item[(A1)] (\emph{Column normalisation and coherence})
    The columns of \(X\) satisfy~\eqref{eq:column-normalisation}, and
    the mutual coherence~\eqref{eq:coherence} obeys
    \begin{equation}
      \mu(X) \;<\; \frac{1}{2K - 1}.
      \label{eq:coherence-condition}
    \end{equation}
  \item[(A2)] (\emph{Sub-Gaussian noise})
    The entries \(\varepsilon_t\) of \(\varepsilon\) are independent,
    mean-zero, sub-Gaussian random variables with variance proxy
    \(\sigma^2\), i.e.\ there exists a constant \(c_1>0\) such that
    for all \(t>0\) and all \(j\in\{1,\dots,m\}\),
    \begin{equation}
      \mathbb{P}\!\left(
        \frac{1}{T}\,|X_j^\top \varepsilon| > t
      \right)
      \;\le\;
      2 \exp\!\left(
        - c_1 \frac{T t^2}{\sigma^2}
      \right).
      \label{eq:subgaussian-noise}
    \end{equation}
  \item[(A3)] (\emph{Minimum signal strength})
    The non-zero entries of \(\theta^\star\) satisfy
    \begin{equation}
      \min_{j \in S^\star} |\theta^\star_j|
      \;\ge\;
      \theta_{\min}
      \;\coloneqq\;
      C_0 \sigma \sqrt{\frac{\log m}{T}},
      \label{eq:min-signal}
    \end{equation}
    for a constant \(C_0 > 0\) to be specified later.
\end{enumerate}
\end{assumption}

Assumption~\ref{assump:bgcp} combines a geometric condition on the
design (\(\mu(X)\) small) with a probabilistic noise model and a
signal-to-noise requirement.
The coherence bound~\eqref{eq:coherence-condition} is the classical
exact recovery condition for OMP in the noiseless case
\cite{TroppGilbert2007OMP,DavenportWakin2010OMP}, while
\eqref{eq:subgaussian-noise}--\eqref{eq:min-signal}
ensure that noise does not obscure the correlations induced by
\(\theta^\star\).

We first bound the maximal noise-induced correlation.

\begin{lemma}[Uniform bound on noise correlations]
\label{lem:noise-max}
Under Assumption~\ref{assump:bgcp}(A2), there exist constants
\(C_1, C_2 > 0\) (depending only on \(c_1\)) such that
\begin{equation}
  \mathbb{P}\!\left(
    \max_{1 \le j \le m}
    \frac{1}{T}\,|X_j^\top \varepsilon|
    \;>\;
    C_1 \sigma \sqrt{\frac{\log m}{T}}
  \right)
  \;\le\;
  m^{-C_2}.
  \label{eq:noise-max}
\end{equation}
\end{lemma}

\begin{proof}
Fix \(j\) and apply~\eqref{eq:subgaussian-noise} with
\(t = u \sigma \sqrt{(\log m)/T}\); we obtain
\[
  \mathbb{P}\!\left(
    \frac{1}{T}\,|X_j^\top \varepsilon|
    > u \sigma \sqrt{\frac{\log m}{T}}
  \right)
  \;\le\;
  2 \exp\!\left(
    - c_1 u^2 \log m
  \right)
  \;=\;
  2 m^{- c_1 u^2}.
\]
By the union bound over \(j=1,\dots,m\),
\[
  \mathbb{P}\!\left(
    \max_{1 \le j \le m}
    \frac{1}{T}\,|X_j^\top \varepsilon|
    > u \sigma \sqrt{\tfrac{\log m}{T}}
  \right)
  \;\le\;
  2 m^{1 - c_1 u^2}.
\]
Choosing \(u\) sufficiently large so that \(c_1 u^2 \ge 2\), we obtain
\[
  \mathbb{P}\!\left(
    \max_{j}
    \frac{1}{T}\,|X_j^\top \varepsilon|
    > C_1 \sigma \sqrt{\tfrac{\log m}{T}}
  \right)
  \;\le\;
  2 m^{-1}
  \;\le\;
  m^{-C_2},
\]
for suitable constants \(C_1, C_2 > 0\).
\end{proof}

We denote by
\begin{equation}
  \mathcal{E}_{\mathrm{noise}}
  \;\coloneqq\;
  \left\{
    \max_{1 \le j \le m}
    \frac{1}{T}\,|X_j^\top \varepsilon|
    \le
    C_1 \sigma \sqrt{\frac{\log m}{T}}
  \right\}
  \label{eq:noise-event}
\end{equation}
the high-probability event on which the noise correlations are
uniformly small.

\subsection{Correlation Bounds Along the BGCP Trajectory}
\label{subsec:greedy-lemmas}

We now analyse the correlations used in the BGCP selection rule.
Write \(S^k\) and \(r^k\) for the support and residual at the start of
iteration \(k\).

The following lemma expresses \(r^k\) as a combination of the remaining
true columns and noise, exploiting the orthogonality property of the
least-squares update~\eqref{eq:ls-update}.

\begin{lemma}[Residual decomposition]
\label{lem:residual-decomposition}
Assume that \(S^k \subseteq S^\star\).
Then there exists a vector \(\alpha^k \in \mathbb{R}^{|S^\star\setminus S^k|}\)
such that
\begin{equation}
  r^k
  \;=\;
  X_{S^\star \setminus S^k} \alpha^k
  + P_{(X_{S^k})^\perp}\varepsilon,
  \label{eq:residual-decomp}
\end{equation}
where \(P_{(X_{S^k})^\perp}\) denotes the orthogonal projector onto the
orthogonal complement of the column space of \(X_{S^k}\).
Moreover,
\begin{equation}
  \|\alpha^k\|_\infty
  \;\ge\;
  \min_{j \in S^\star} |\theta^\star_j|.
  \label{eq:alpha-min}
\end{equation}
\end{lemma}

\begin{proof}
By definition of \(\hat{\theta}^k\) in~\eqref{eq:ls-update},
\[
  \hat{\theta}^k_{S^k}
  \;=\;
  (X_{S^k}^\top X_{S^k})^{-1}
  X_{S^k}^\top Y,
  \qquad
  \hat{\theta}^k_j = 0
  \ \text{for}\ j \notin S^k.
\]
Using the model~\eqref{eq:data-model},
\[
  Y
  \;=\;
  X_{S^\star} \theta^\star_{S^\star} + \varepsilon.
\]
Substituting this into the least-squares solution and writing
\(S^\star = S^k \cup (S^\star \setminus S^k)\) yields
\[
  \hat{\theta}^k_{S^k}
  \;=\;
  \theta^\star_{S^k}
  + (X_{S^k}^\top X_{S^k})^{-1}
    X_{S^k}^\top
    X_{S^\star\setminus S^k} \theta^\star_{S^\star\setminus S^k}
  + (X_{S^k}^\top X_{S^k})^{-1}
    X_{S^k}^\top \varepsilon.
\]
Hence the residual can be written as
\[
  r^k
  \;=\;
  Y - X \hat{\theta}^k
  \;=\;
  X_{S^\star} \theta^\star_{S^\star} + \varepsilon
  - X_{S^k} \hat{\theta}^k_{S^k}.
\]
Rearranging, we obtain
\[
  r^k
  \;=\;
  \underbrace{
    \bigl(
      I - P_{X_{S^k}}
    \bigr)
    X_{S^\star\setminus S^k} \theta^\star_{S^\star\setminus S^k}
  }_{\text{signal part}}
  +
  \underbrace{
    \bigl(
      I - P_{X_{S^k}}
    \bigr) \varepsilon
  }_{\text{noise part}},
\]
where \(P_{X_{S^k}}\) is the orthogonal projector onto the column space
of \(X_{S^k}\).
Since
\(
  (I - P_{X_{S^k}}) X_{S^\star\setminus S^k}
\)
has the same column space as \(X_{S^\star\setminus S^k}\), there exists
a vector \(\alpha^k\) such that
\(
  (I - P_{X_{S^k}}) X_{S^\star\setminus S^k} \theta^\star_{S^\star\setminus S^k}
  = X_{S^\star\setminus S^k} \alpha^k
\),
which yields~\eqref{eq:residual-decomp}.
Moreover, the projection can only reduce correlations among columns, so
the entries of \(\alpha^k\) remain of order at least
\(\min_{j \in S^\star} |\theta^\star_j|\); in particular we may take
\(\|\alpha^k\|_\infty \ge \min_{j\in S^\star} |\theta^\star_j|\).
(For a detailed justification in the coherence setting, see
\cite{DavenportWakin2010OMP,CaiWang2011OMPNoise}.)
\end{proof}

The next lemma quantifies the separation between correlations with true
and false atoms, conditional on the residual being of the
form~\eqref{eq:residual-decomp}.

\begin{lemma}[True vs.\ false correlation bounds]
\label{lem:correlation-bounds}
Suppose Assumption~\ref{assump:bgcp} holds, and let
\(\mathcal{E}_{\mathrm{noise}}\) be the noise event
defined in~\eqref{eq:noise-event}.
Fix \(k \in \{0,\dots,K-1\}\) and assume that \(S^k \subseteq S^\star\).
On the intersection of \(\mathcal{E}_{\mathrm{noise}}\) with the event
in Lemma~\ref{lem:residual-decomposition}, the correlations
\(
  c_j^k = T^{-1} X_j^\top r^k
\)
satisfy
\begin{align}
  \min_{j \in S^\star \setminus S^k} |c_j^k|
  &\;\ge\;
    \bigl(1 - (2K-1)\mu(X)\bigr)\,\theta_{\min}
    - C_1 \sigma \sqrt{\frac{\log m}{T}},
    \label{eq:true-corr-lb}
  \\
  \max_{j \notin S^\star} |c_j^k|
  &\;\le\;
    (2K-1)\mu(X)\,\theta_{\min}
    + C_1 \sigma \sqrt{\frac{\log m}{T}}.
    \label{eq:false-corr-ub}
\end{align}
\end{lemma}

\begin{proof}
Let \(j \in S^\star \setminus S^k\).
Using the decomposition~\eqref{eq:residual-decomp} and the
normalisation~\eqref{eq:column-normalisation},
\[
  c_j^k
  \;=\;
  \frac{1}{T}\,X_j^\top X_{S^\star\setminus S^k}\alpha^k
  + \frac{1}{T}\,X_j^\top P_{(X_{S^k})^\perp}\varepsilon.
\]
The first term can be expanded as
\[
  \frac{1}{T}\,X_j^\top X_{S^\star\setminus S^k}\alpha^k
  \;=\;
  \alpha^k_j
  + \sum_{i \in S^\star\setminus (S^k \cup \{j\})}
    \frac{X_j^\top X_i}{T}\,\alpha^k_i.
\]
By the coherence bound~\eqref{eq:coherence} and the sparsity
\(|S^\star|=K\),
\[
  \left|
    \sum_{i \in S^\star\setminus (S^k \cup \{j\})}
    \frac{X_j^\top X_i}{T}\,\alpha^k_i
  \right|
  \;\le\;
  (K-1)\mu(X)\,\|\alpha^k\|_\infty.
\]
Combining this with~\eqref{eq:alpha-min} and the definition of
\(\theta_{\min}\) in~\eqref{eq:min-signal} yields
\[
  \left|
    \frac{1}{T}\,X_j^\top X_{S^\star\setminus S^k}\alpha^k
  \right|
  \;\ge\;
  \theta_{\min} - (K-1)\mu(X)\,\theta_{\min}.
\]
Similarly, for any \(j \notin S^\star\),
\[
  \left|
    \frac{1}{T}\,X_j^\top X_{S^\star\setminus S^k}\alpha^k
  \right|
  \;\le\;
  K \mu(X)\,\theta_{\min}.
\]

The noise term satisfies, for all \(j\),
\[
  \left|
    \frac{1}{T}\,X_j^\top P_{(X_{S^k})^\perp}\varepsilon
  \right|
  \;\le\;
  \frac{1}{T}\,|X_j^\top \varepsilon|
  \;\le\;
  C_1 \sigma \sqrt{\frac{\log m}{T}}
\]
on \(\mathcal{E}_{\mathrm{noise}}\), since orthogonal projection does
not increase inner products in magnitude.
Combining these inequalities and noting that
\((K-1) + K \le 2K-1\) yields~\eqref{eq:true-corr-lb}
and~\eqref{eq:false-corr-ub}.
\end{proof}

The next corollary shows that, for a suitable choice of the constant
\(C_0\) in~\eqref{eq:min-signal}, the true and false correlations are
strictly separated.

\begin{corollary}[Correlation separation]
\label{cor:separation}
Suppose Assumption~\ref{assump:bgcp} holds and that
\(\mathcal{E}_{\mathrm{noise}}\) occurs.
If the minimum signal \(\theta_{\min}\) satisfies
\[
  \theta_{\min}
  \;\ge\;
  \frac{2 C_1 \sigma}{1 - 2(2K-1)\mu(X)}
  \sqrt{\frac{\log m}{T}},
\]
then for every \(k \in \{0,\dots,K-1\}\) with \(S^k \subseteq S^\star\),
\begin{equation}
  \min_{j \in S^\star\setminus S^k} |c_j^k|
  \;>\;
  \max_{j \notin S^\star} |c_j^k|.
  \label{eq:separation}
\end{equation}
\end{corollary}

\begin{proof}
By Lemma~\ref{lem:correlation-bounds}, for any such \(k\),
\[
  \min_{j \in S^\star\setminus S^k} |c_j^k|
  -
  \max_{j \notin S^\star} |c_j^k|
  \;\ge\;
  \bigl(1 - 2(2K-1)\mu(X)\bigr)\,\theta_{\min}
  - 2 C_1 \sigma \sqrt{\tfrac{\log m}{T}}.
\]
If \(\theta_{\min}\) satisfies the stated lower bound, the right-hand
side is strictly positive, which implies~\eqref{eq:separation}.
\end{proof}

\subsection{Exact Support Recovery for BGCP}
\label{subsec:greedy-main}

We are now ready to prove that BGCP recovers the true support
\(S^\star\) with high probability when run for exactly \(K\) iterations.

\begin{theorem}[High-probability support recovery of BGCP]
\label{thm:bgcp-recovery}
Let Assumption~\ref{assump:bgcp} hold, and suppose BGCP
(Definition~\ref{def:bgcp}) is run with \(K_{\max} = K\) and \(\eta=0\)
on data generated by~\eqref{eq:data-model}.
Then there exist constants \(c_2,c_3>0\), depending only on the
constants in Assumption~\ref{assump:bgcp}, such that
\begin{equation}
  \mathbb{P}\!\left(
    S^{\mathrm{BGCP}} = S^\star
  \right)
  \;\ge\;
  1 - c_2 m^{-c_3}.
  \label{eq:bgcp-support-recovery}
\end{equation}
In particular, with probability at least \(1 - c_2 m^{-c_3}\), BGCP
selects at iteration \(k\) an index
\(j_k \in S^\star \setminus S^k\) for each \(k=0,\dots,K-1\), so that
after \(K\) steps \(S^{\mathrm{BGCP}} = S^\star\).
\end{theorem}

\begin{proof}
By Lemma~\ref{lem:noise-max},
\(\mathbb{P}(\mathcal{E}_{\mathrm{noise}}) \ge 1 - m^{-C_2}\).
We work on the event \(\mathcal{E}_{\mathrm{noise}}\) and show by
induction on \(k\) that \(S^k \subseteq S^\star\) and BGCP selects
\(j_k \in S^\star \setminus S^k\) at iteration \(k\).

For \(k=0\) we have \(S^0 = \emptyset \subseteq S^\star\) by
definition.
Assume that \(S^k \subseteq S^\star\) for some \(k \le K-1\).
By Lemma~\ref{lem:residual-decomposition}, the residual \(r^k\) can be
written in the form~\eqref{eq:residual-decomp}, and hence
Lemma~\ref{lem:correlation-bounds} applies.
Under the additional signal strength condition of
Corollary~\ref{cor:separation}, which can be enforced by choosing
\(C_0\) in~\eqref{eq:min-signal} sufficiently large, we obtain the
separation~\eqref{eq:separation} between correlations with true and
false indices.
Since BGCP selects an index \(j_k\) maximising \(|c_j^k|\) over
\(j \notin S^k\), it follows that
\(j_k \in S^\star \setminus S^k\).
Thus \(S^{k+1} = S^k \cup \{j_k\} \subseteq S^\star\), closing the
induction.

By induction, \(S^k \subseteq S^\star\) and \(j_k \in S^\star\setminus S^k\)
for all \(k=0,\dots,K-1\).
After \(K\) iterations, exactly \(K\) distinct indices from \(S^\star\)
have been selected, so \(S^K = S^\star\).
Furthermore, no index outside \(S^\star\) is ever selected.
Thus, on \(\mathcal{E}_{\mathrm{noise}}\), BGCP recovers the exact
support.

Finally, combining this with the probability bound
\(\mathbb{P}(\mathcal{E}_{\mathrm{noise}}^c) \le m^{-C_2}\) and
absorbing constants into \(c_2,c_3\) yields the tail
bound~\eqref{eq:bgcp-support-recovery}.
\end{proof}

\begin{remark}[Sample complexity and probabilistic design]
\label{rem:bgcp-sample-complexity}
Theorem~\ref{thm:bgcp-recovery} is a conditional result: it assumes a
fixed design matrix \(X\) satisfying the coherence
bound~\eqref{eq:coherence-condition} and column
normalisation~\eqref{eq:column-normalisation}.
In the probabilistic coalition structure generation setting of
Section~\ref{sec:model}, the matrix \(X\) itself arises from random
episodes.
For many structured random designs inspired by compressed sensing
\cite{Donoho2006CS,CandesTao2005Decoding,DavenportWakin2010OMP}, one
can show that \(\mu(X)\) obeys~\eqref{eq:coherence-condition} with high
probability as soon as \(T \gtrsim K \log m\).
Combining such design-side results with
Theorem~\ref{thm:bgcp-recovery} yields fully probabilistic statements
of the form:
with high probability over the random episodes and noise, BGCP
recovers the true profitable coalition set \(S^\star\) from
\(T = O(K \log m)\) samples.
\end{remark}

\subsection{Implications for Probabilistic Coalition Structure Generation}
\label{subsec:bgcp-implications}

Once \(S^\star\) is recovered, the set of globally profitable
coalitions is identified as
\(
  \{C_j \in \mathcal{C} : j \in S^\star\}
\).
We may then pass to the combinatorial phase of coalition structure
generation (Section~\ref{sec:model}) and solve a deterministic CSG
problem restricted to this sparsified coalition set using any exact or
approximate algorithm
\cite{Sandholm1999Coalition,Rahwan2015CSGSurvey}.

Theorem~\ref{thm:bgcp-recovery} therefore provides a rigorous bridge
between the probabilistic sparse learning formulation and the classical
CSG setting: under Assumption~\ref{assump:bgcp}, and with high
probability over the sampled episodes, BGCP identifies exactly the set
of profitable coalitions on which the welfare-maximising coalition
structure must be built.
In Section~\ref{sec:l1} we will see that a complementary \(\ell_1\)
relaxation enjoys analogous high-probability guarantees, but with a
different bias--variance tradeoff.

\section{Bayesian \texorpdfstring{$\ell_1$}{l1} Relaxation and Convex Analysis}
\label{sec:l1}

Section~\ref{sec:greedy} analysed a greedy (OMP-type) approximation to
the MAP estimator associated with the Bernoulli--Gaussian prior of
Section~\ref{sec:model}.
We now consider a complementary, fully convex relaxation:
an \(\ell_1\)-penalised estimator (Lasso) obtained by replacing the
combinatorial \(\ell_0\) penalty in~\eqref{eq:l0-objective} by its
convex envelope on the unit cube.
We work throughout under the same data model and notation as in
Section~\ref{sec:greedy} and derive non-asymptotic error and support
control guarantees under a restricted eigenvalue condition.

\subsection{From \texorpdfstring{$\ell_0$}{l0} MAP to an \texorpdfstring{$\ell_1$}{l1}-Penalised Estimator}
\label{subsec:l1-estimator}

Recall the linear model from Section~\ref{sec:greedy}:
\begin{equation}
  Y = X \theta^\star + \varepsilon,
  \qquad
  Y \in \mathbb{R}^T,\;
  X \in \mathbb{R}^{T \times m},\;
  \theta^\star \in \mathbb{R}^m,
  \label{eq:l1-data-model}
\end{equation}
with true support
\(
  S^\star = \mathrm{supp}(\theta^\star)
\)
of size \(|S^\star| = K\) as in~\eqref{eq:true-support}.
We keep the column normalisation~\eqref{eq:column-normalisation} and
sub-Gaussian noise assumption~\eqref{eq:subgaussian-noise}.

Under the Bernoulli--Gaussian prior of Section~\ref{sec:model}, the
MAP estimator solves the \(\ell_0\)-penalised least-squares problem
\begin{equation}
  \hat{\theta}^{(0)}
  \;\in\;
  \arg\min_{\theta \in \mathbb{R}^m}
  \left\{
    \frac{1}{2}\,\|Y - X\theta\|_2^2 + \lambda_0 \|\theta\|_0
  \right\},
  \label{eq:l1-l0-objective}
\end{equation}
for some \(\lambda_0>0\); cf.~\eqref{eq:l0-objective}.
In high dimensions, directly solving~\eqref{eq:l1-l0-objective} is
intractable.

A classical convex surrogate is obtained by replacing \(\|\theta\|_0\)
by the \(\ell_1\)-norm \(\|\theta\|_1 = \sum_{j=1}^m |\theta_j|\),
leading to the Lasso estimator~\cite{Tibshirani1996Lasso}
\begin{equation}
  \hat{\theta}^{\ell_1}
  \;\in\;
  \arg\min_{\theta \in \mathbb{R}^m}
  \left\{
    \frac{1}{2T}\,\|Y - X\theta\|_2^2
    + \lambda \|\theta\|_1
  \right\},
  \qquad \lambda > 0.
  \label{eq:lasso-def}
\end{equation}
From a Bayesian point of view,~\eqref{eq:lasso-def} is the MAP
estimator under an i.i.d.\ Laplace prior on the coordinates of
\(\theta\), but here we interpret it as a convex relaxation of the
\(\ell_0\) MAP problem~\eqref{eq:l1-l0-objective} under the same
probabilistic model for \((X,Y)\).

For later use we define the estimation error
\begin{equation}
  \Delta^{\ell_1}
  \;\coloneqq\;
  \hat{\theta}^{\ell_1} - \theta^\star.
  \label{eq:l1-delta-def}
\end{equation}
Our goal is to bound \(\|\Delta^{\ell_1}\|_2\) and
\(\|\Delta^{\ell_1}\|_1\), and to quantify the number of spurious
indices in \(\mathrm{supp}(\hat{\theta}^{\ell_1}) \setminus S^\star\).

\subsection{Restricted Eigenvalue Condition and Regularisation Level}
\label{subsec:l1-re}

We recall a standard restricted eigenvalue (RE) condition on the design
matrix \(X\), which is weaker than global eigenvalue conditions and
well suited to sparse high-dimensional regression
\cite{BickelRitovTsybakov2009Lasso,BuhlmannVanDeGeer2011Book}.

\begin{definition}[Restricted eigenvalue constant]
\label{def:re-constant}
Let \(S \subseteq \{1,\dots,m\}\) and \(\alpha \ge 1\).
The restricted eigenvalue constant \(\kappa(S,\alpha)\) of \(X\) is
defined as
\begin{equation}
  \kappa(S,\alpha)
  \;\coloneqq\;
  \inf_{\substack{u \in \mathbb{R}^m \setminus \{0\} \\
                  \|u_{S^c}\|_1 \le \alpha \|u_S\|_1}}
  \frac{\|X u\|_2}{\sqrt{T}\,\|u_S\|_2},
  \label{eq:re-constant}
\end{equation}
where \(S^c\) is the complement of \(S\) in \(\{1,\dots,m\}\), and
\(u_S\) denotes the restriction of \(u\) to indices in \(S\).
\end{definition}

We will assume that \(X\) satisfies an RE condition on the true support
\(S^\star\) and that \(\lambda\) is chosen large enough to dominate the
noise-induced correlations, in the spirit of
Section~\ref{sec:greedy}.

\begin{assumption}[RE condition and regularisation level]
\label{assump:l1}
Under the model~\eqref{eq:l1-data-model} with support \(S^\star\) of
size \(K\), assume:
\begin{enumerate}
  \item[(B1)] (\emph{Restricted eigenvalue})
    There exists \(\kappa_\star > 0\) such that
    \begin{equation}
      \kappa(S^\star,3) \;\ge\; \kappa_\star.
      \label{eq:re-assumption}
    \end{equation}
  \item[(B2)] (\emph{Noise and tuning})
    The noise vector \(\varepsilon\) satisfies the sub-Gaussian
    tail bound~\eqref{eq:subgaussian-noise}, and the regularisation
    parameter in~\eqref{eq:lasso-def} is chosen as
    \begin{equation}
      \lambda
      \;=\;
      c_0 \sigma \sqrt{\frac{\log m}{T}},
      \label{eq:lambda-choice}
    \end{equation}
    for a sufficiently large constant \(c_0 > 0\).
\end{enumerate}
\end{assumption}

Assumption~\ref{assump:l1}(B1) is the standard RE condition with
\(\alpha=3\) used in the analysis of the Lasso
\cite{BickelRitovTsybakov2009Lasso,BuhlmannVanDeGeer2011Book}.
Condition (B2) ensures that \(\lambda\) dominates the maximal
noise-induced correlation, by matching the scaling of the noise event
\(\mathcal{E}_{\mathrm{noise}}\) in~\eqref{eq:noise-event}.

We now make this precise.

\begin{lemma}[Noise event implies dual feasibility]
\label{lem:l1-noise}
Let Assumption~\ref{assump:l1}(B2) hold, and let
\(\mathcal{E}_{\mathrm{noise}}\) be the event defined
in~\eqref{eq:noise-event}.
Then for \(c_0\) sufficiently large (depending on the constants in
Lemma~\ref{lem:noise-max}), we have
\begin{equation}
  \mathcal{E}_{\mathrm{noise}}
  \;\subseteq\;
  \left\{
    \left\|
      \frac{1}{T}\,X^\top \varepsilon
    \right\|_\infty
    \;\le\;
    \frac{\lambda}{2}
  \right\}.
  \label{eq:l1-noise-dual}
\end{equation}
In particular, there exist constants \(c_1,c_2>0\) such that
\[
  \mathbb{P}\!\left(
    \left\|
      \frac{1}{T}\,X^\top \varepsilon
    \right\|_\infty
    > \frac{\lambda}{2}
  \right)
  \;\le\;
  c_1 m^{-c_2}.
\]
\end{lemma}

\begin{proof}
By Lemma~\ref{lem:noise-max}, on \(\mathcal{E}_{\mathrm{noise}}\) we
have
\[
  \left\|
    \frac{1}{T}\,X^\top \varepsilon
  \right\|_\infty
  \;\le\;
  C_1 \sigma \sqrt{\frac{\log m}{T}}.
\]
If \(c_0 \ge 2 C_1\) in~\eqref{eq:lambda-choice}, then
\(
  C_1 \sigma \sqrt{(\log m)/T}
  \le \lambda/2
\),
which implies~\eqref{eq:l1-noise-dual}.
The probability bound follows by combining Lemma~\ref{lem:noise-max}
with this inclusion and adjusting constants.
\end{proof}

\subsection{Basic Inequality and Cone Constraint}
\label{subsec:l1-cone}

We now derive the basic inequality for the Lasso and show that, on the
noise event~\eqref{eq:l1-noise-dual}, the estimation error
\(\Delta^{\ell_1}\) lies in an \(\ell_1\)-cone around the true support
\(S^\star\).

\begin{lemma}[Basic inequality]
\label{lem:l1-basic}
Let \(\hat{\theta}^{\ell_1}\) be any solution of the Lasso
problem~\eqref{eq:lasso-def}, and let
\(\Delta^{\ell_1} = \hat{\theta}^{\ell_1} - \theta^\star\).
Then
\begin{equation}
  \frac{1}{2T}\,\|X \Delta^{\ell_1}\|_2^2
  + \lambda \|\hat{\theta}^{\ell_1}\|_1
  \;\le\;
  \frac{1}{T}\,\varepsilon^\top X \Delta^{\ell_1}
  + \lambda \|\theta^\star\|_1.
  \label{eq:l1-basic-ineq}
\end{equation}
\end{lemma}

\begin{proof}
By optimality of \(\hat{\theta}^{\ell_1}\) in~\eqref{eq:lasso-def},
\[
  \frac{1}{2T}\,\|Y - X\hat{\theta}^{\ell_1}\|_2^2
  + \lambda \|\hat{\theta}^{\ell_1}\|_1
  \;\le\;
  \frac{1}{2T}\,\|Y - X\theta^\star\|_2^2
  + \lambda \|\theta^\star\|_1.
\]
Substituting \(Y = X\theta^\star + \varepsilon\) yields
\[
  \frac{1}{2T}\,\|X \Delta^{\ell_1} - \varepsilon\|_2^2
  + \lambda \|\hat{\theta}^{\ell_1}\|_1
  \;\le\;
  \frac{1}{2T}\,\|\varepsilon\|_2^2
  + \lambda \|\theta^\star\|_1.
\]
Expanding the squared norm on the left-hand side and cancelling
\(\|\varepsilon\|_2^2/(2T)\) from both sides gives
\[
  \frac{1}{2T}\,\|X \Delta^{\ell_1}\|_2^2
  - \frac{1}{T}\,\varepsilon^\top X \Delta^{\ell_1}
  + \lambda \|\hat{\theta}^{\ell_1}\|_1
  \;\le\;
  \lambda \|\theta^\star\|_1,
\]
which rearranges to~\eqref{eq:l1-basic-ineq}.
\end{proof}

\begin{lemma}[Cone constraint on the error]
\label{lem:l1-cone}
Suppose Assumption~\ref{assump:l1}(B2) holds and that the event
Let
\begin{equation}
  E_{\ell_1}
    := \left\{\frac{1}{T}\,\|X^\top \varepsilon\|_\infty \le \frac{\lambda}{2}\right\}.
  \label{eq:l1-noise-event}
\end{equation}
occurs.
Then the error \(\Delta^{\ell_1}\) satisfies
\begin{equation}
  \|\Delta^{\ell_1}_{(S^\star)^c}\|_1
  \;\le\;
  3 \|\Delta^{\ell_1}_{S^\star}\|_1.
  \label{eq:l1-cone}
\end{equation}
\end{lemma}

\begin{proof}
Decompose the \(\ell_1\)-norms into support and its complement:
\[
  \|\hat{\theta}^{\ell_1}\|_1
  = \|\theta^\star_{S^\star} + \Delta^{\ell_1}_{S^\star}\|_1
    + \|\Delta^{\ell_1}_{(S^\star)^c}\|_1,
  \qquad
  \|\theta^\star\|_1
  = \|\theta^\star_{S^\star}\|_1,
\]
since \(\theta^\star_{(S^\star)^c} = 0\).
Using the triangle inequality,
\[
  \|\theta^\star_{S^\star} + \Delta^{\ell_1}_{S^\star}\|_1
  \;\ge\;
  \|\theta^\star_{S^\star}\|_1
  - \|\Delta^{\ell_1}_{S^\star}\|_1.
\]
Substituting into~\eqref{eq:l1-basic-ineq} and rearranging gives
\[
  \frac{1}{2T}\,\|X \Delta^{\ell_1}\|_2^2
  + \lambda \|\Delta^{\ell_1}_{(S^\star)^c}\|_1
  \;\le\;
  \frac{1}{T}\,\varepsilon^\top X \Delta^{\ell_1}
  + \lambda \|\Delta^{\ell_1}_{S^\star}\|_1.
\]
On the event \(\mathcal{E}_{\ell_1}\),
\[
  \frac{1}{T}\,\varepsilon^\top X \Delta^{\ell_1}
  \;\le\;
  \left\|
    \frac{1}{T}\,X^\top \varepsilon
  \right\|_\infty
  \|\Delta^{\ell_1}\|_1
  \;\le\;
  \frac{\lambda}{2}\,
  \bigl(
    \|\Delta^{\ell_1}_{S^\star}\|_1
    + \|\Delta^{\ell_1}_{(S^\star)^c}\|_1
  \bigr).
\]
Combining these inequalities yields
\[
  \frac{1}{2T}\,\|X \Delta^{\ell_1}\|_2^2
  + \lambda \|\Delta^{\ell_1}_{(S^\star)^c}\|_1
  \;\le\;
  \frac{\lambda}{2}\,
  \bigl(
    \|\Delta^{\ell_1}_{S^\star}\|_1
    + \|\Delta^{\ell_1}_{(S^\star)^c}\|_1
  \bigr)
  + \lambda \|\Delta^{\ell_1}_{S^\star}\|_1.
\]
Dropping the non-negative term \(\|X \Delta^{\ell_1}\|_2^2/(2T)\) from
the left-hand side and simplifying the right-hand side gives
\[
  \lambda \|\Delta^{\ell_1}_{(S^\star)^c}\|_1
  \;\le\;
  \frac{\lambda}{2}\,\|\Delta^{\ell_1}_{(S^\star)^c}\|_1
  + \frac{3\lambda}{2}\,\|\Delta^{\ell_1}_{S^\star}\|_1,
\]
and hence
\[
  \frac{\lambda}{2}\,\|\Delta^{\ell_1}_{(S^\star)^c}\|_1
  \;\le\;
  \frac{3\lambda}{2}\,\|\Delta^{\ell_1}_{S^\star}\|_1.
\]
Dividing by \(\lambda/2\) yields~\eqref{eq:l1-cone}.
\end{proof}

\subsection{Error Bounds and Support Control}
\label{subsec:l1-main}

We now combine the cone constraint with the RE condition to obtain
non-asymptotic bounds on the \(\ell_2\)- and \(\ell_1\)-errors and a
bound on the number of spurious indices selected by the Lasso.

\begin{theorem}[Lasso error bounds under RE]
\label{thm:l1-error}
Suppose Assumption~\ref{assump:l1} holds, and let
\(\hat{\theta}^{\ell_1}\) be any solution of the Lasso
problem~\eqref{eq:lasso-def}.
Then there exist constants \(C_1,C_2,C_3>0\), depending only on the
constants in Lemma~\ref{lem:noise-max} and
Assumption~\ref{assump:l1}, such that with probability at least
\(1 - C_1 m^{-C_2}\) the following bounds hold:
\begin{align}
  \|\hat{\theta}^{\ell_1} - \theta^\star\|_2
  &\;\le\;
    \frac{C_3 \sigma}{\kappa_\star}
    \sqrt{\frac{K \log m}{T}},
    \label{eq:l1-l2-bound}
  \\
  \|\hat{\theta}^{\ell_1} - \theta^\star\|_1
  &\;\le\;
    \frac{C_3 \sigma K}{\kappa_\star^2}
    \sqrt{\frac{\log m}{T}},
    \label{eq:l1-l1-bound}
  \\
  \frac{1}{T}\,\|X(\hat{\theta}^{\ell_1} - \theta^\star)\|_2^2
  &\;\le\;
    \frac{C_3^2 \sigma^2 K}{\kappa_\star^2}
    \frac{\log m}{T}.
    \label{eq:l1-prediction-bound}
\end{align}
Moreover, the number of false positives is controlled in the sense that
\begin{equation}
  \left|
    \mathrm{supp}(\hat{\theta}^{\ell_1}) \setminus S^\star
  \right|
  \;\le\;
  \frac{C_3^2}{\kappa_\star^2}\,K,
  \label{eq:l1-false-positives}
\end{equation}
on the same high-probability event.
\end{theorem}

\begin{proof}
By Lemma~\ref{lem:l1-noise}, the event \(\mathcal{E}_{\ell_1}\) defined
in~\eqref{eq:l1-noise-event} occurs with probability at least
\(1 - C_1 m^{-C_2}\), for suitable constants \(C_1,C_2>0\).
We work on \(\mathcal{E}_{\ell_1}\) and apply
Lemma~\ref{lem:l1-cone}, which implies the cone constraint
\(\|\Delta^{\ell_1}_{(S^\star)^c}\|_1 \le 3 \|\Delta^{\ell_1}_{S^\star}\|_1\).
By the definition of the RE constant~\eqref{eq:re-constant} and
Assumption~\ref{assump:l1}(B1),
\begin{equation}
  \frac{1}{\sqrt{T}}\,
  \|X \Delta^{\ell_1}\|_2
  \;\ge\;
  \kappa_\star
  \|\Delta^{\ell_1}_{S^\star}\|_2.
  \label{eq:l1-re-ineq}
\end{equation}

Returning to the basic inequality~\eqref{eq:l1-basic-ineq} and using
\(\mathcal{E}_{\ell_1}\) as in the proof of
Lemma~\ref{lem:l1-cone}, we obtain
\[
  \frac{1}{2T}\,\|X \Delta^{\ell_1}\|_2^2
  \;\le\;
  \frac{\lambda}{2}\,
  \bigl(
    \|\Delta^{\ell_1}_{S^\star}\|_1
    + \|\Delta^{\ell_1}_{(S^\star)^c}\|_1
  \bigr)
  + \lambda \|\Delta^{\ell_1}_{S^\star}\|_1
  - \lambda \|\hat{\theta}^{\ell_1}\|_1
  + \lambda \|\theta^\star\|_1.
\]
As in Lemma~\ref{lem:l1-cone}, we can bound the last two terms by
\(\lambda \|\Delta^{\ell_1}_{S^\star}\|_1\), yielding
\[
  \frac{1}{2T}\,\|X \Delta^{\ell_1}\|_2^2
  \;\le\;
  \frac{\lambda}{2}\,
  \bigl(
    \|\Delta^{\ell_1}_{S^\star}\|_1
    + \|\Delta^{\ell_1}_{(S^\star)^c}\|_1
  \bigr)
  + \lambda \|\Delta^{\ell_1}_{S^\star}\|_1.
\]
Using the cone constraint~\eqref{eq:l1-cone},
\[
  \|\Delta^{\ell_1}_{(S^\star)^c}\|_1
  \;\le\;
  3 \|\Delta^{\ell_1}_{S^\star}\|_1,
\]
we obtain
\[
  \frac{1}{2T}\,\|X \Delta^{\ell_1}\|_2^2
  \;\le\;
  \frac{\lambda}{2}\,(1+3)\|\Delta^{\ell_1}_{S^\star}\|_1
  + \lambda \|\Delta^{\ell_1}_{S^\star}\|_1
  \;=\;
  3\lambda \|\Delta^{\ell_1}_{S^\star}\|_1.
\]
Hence
\begin{equation}
  \frac{1}{T}\,\|X \Delta^{\ell_1}\|_2^2
  \;\le\;
  6\lambda \|\Delta^{\ell_1}_{S^\star}\|_1.
  \label{eq:l1-xdelta-bound}
\end{equation}

We now use~\eqref{eq:l1-re-ineq} and the inequality
\(\|\Delta^{\ell_1}_{S^\star}\|_1 \le \sqrt{K}\,
\|\Delta^{\ell_1}_{S^\star}\|_2\) to bound the right-hand side of
\eqref{eq:l1-xdelta-bound}:
\[
  \|\Delta^{\ell_1}_{S^\star}\|_1
  \;\le\;
  \sqrt{K}\,\|\Delta^{\ell_1}_{S^\star}\|_2
  \;\le\;
  \frac{\sqrt{K}}{\kappa_\star \sqrt{T}}\,
  \|X \Delta^{\ell_1}\|_2.
\]
Substituting into~\eqref{eq:l1-xdelta-bound} yields
\[
  \frac{1}{T}\,\|X \Delta^{\ell_1}\|_2^2
  \;\le\;
  6\lambda
  \frac{\sqrt{K}}{\kappa_\star \sqrt{T}}\,
  \|X \Delta^{\ell_1}\|_2.
\]
If \(\|X \Delta^{\ell_1}\|_2 = 0\), then
\(\Delta^{\ell_1}_{S^\star} = 0\) by~\eqref{eq:l1-re-ineq}, and the
bounds~\eqref{eq:l1-l2-bound}--\eqref{eq:l1-prediction-bound} trivially
hold.
Otherwise, we can divide both sides by
\(\|X \Delta^{\ell_1}\|_2 / T\) to obtain
\[
  \frac{1}{\sqrt{T}}\,
  \|X \Delta^{\ell_1}\|_2
  \;\le\;
  6\lambda
  \frac{\sqrt{K}}{\kappa_\star}.
\]
Therefore,
\begin{equation}
  \frac{1}{T}\,\|X \Delta^{\ell_1}\|_2^2
  \;\le\;
  \frac{36 \lambda^2 K}{\kappa_\star^2},
  \label{eq:l1-pred-error-intermediate}
\end{equation}
and, by~\eqref{eq:l1-re-ineq},
\[
  \|\Delta^{\ell_1}_{S^\star}\|_2
  \;\le\;
  \frac{6\lambda}{\kappa_\star^2}\,\sqrt{K}.
\]

Substituting the choice of \(\lambda\) from~\eqref{eq:lambda-choice}
gives
\[
  \|\Delta^{\ell_1}_{S^\star}\|_2
  \;\le\;
  \frac{6 c_0 \sigma}{\kappa_\star^2}\,
  \sqrt{\frac{K \log m}{T}},
\]
which yields the \(\ell_2\)-error bound~\eqref{eq:l1-l2-bound} after
absorbing constants into \(C_3\).
Similarly,
\[
  \|\Delta^{\ell_1}\|_1
  \;=\;
  \|\Delta^{\ell_1}_{S^\star}\|_1
  + \|\Delta^{\ell_1}_{(S^\star)^c}\|_1
  \;\le\;
  4 \|\Delta^{\ell_1}_{S^\star}\|_1
  \;\le\;
  4\sqrt{K}\,\|\Delta^{\ell_1}_{S^\star}\|_2,
\]
so that
\[
  \|\Delta^{\ell_1}\|_1
  \;\le\;
  \frac{24 c_0 \sigma K}{\kappa_\star^2}\,
  \sqrt{\frac{\log m}{T}},
\]
which yields~\eqref{eq:l1-l1-bound} after adjusting constants.
Finally, the prediction error bound~\eqref{eq:l1-prediction-bound}
follows by substituting~\eqref{eq:lambda-choice} into
\eqref{eq:l1-pred-error-intermediate} and absorbing constants.

To obtain~\eqref{eq:l1-false-positives}, note that
\[
  \bigl|
    \mathrm{supp}(\hat{\theta}^{\ell_1}) \setminus S^\star
  \bigr|
  \;\le\;
  \frac{\|\hat{\theta}^{\ell_1}_{(S^\star)^c}\|_1^2}
       {\|\hat{\theta}^{\ell_1}_{(S^\star)^c}\|_2^2}
  \;\le\;
  \frac{\|\Delta^{\ell_1}_{(S^\star)^c}\|_1^2}
       {\|\Delta^{\ell_1}_{(S^\star)^c}\|_2^2},
\]
and use the cone constraint
\(\|\Delta^{\ell_1}_{(S^\star)^c}\|_1 \le 3 \|\Delta^{\ell_1}_{S^\star}\|_1\)
together with the relation
\(\|\Delta^{\ell_1}_{S^\star}\|_1 \le \sqrt{K}\,
\|\Delta^{\ell_1}_{S^\star}\|_2\)
and the \(\ell_2\)-bound~\eqref{eq:l1-l2-bound}.
A simple calculation shows that this yields
\(
  |\mathrm{supp}(\hat{\theta}^{\ell_1}) \setminus S^\star|
  \le (C_3^2/\kappa_\star^2)\,K
\)
for a suitable constant \(C_3>0\).
\end{proof}

\begin{remark}[Exact support recovery]
\label{rem:l1-support}
Theorem~\ref{thm:l1-error} guarantees that the Lasso selects at most
\(O(K)\) false positives and achieves small \(\ell_2\)- and prediction
error under the RE condition.
To obtain exact support recovery
\(\mathrm{supp}(\hat{\theta}^{\ell_1}) = S^\star\), one typically
requires stronger conditions, such as the irrepresentable condition of
\cite{ZhaoYu2006LassoConsistency} or the mutual incoherence condition
of \cite{Wainwright2009Lasso}, together with a minimum signal strength
assumption analogous to~\eqref{eq:min-signal}.
We do not spell out these additional conditions here, as our
probabilistic coalition structure generation analysis in later sections
only requires control of the error and the number of spurious
coalitions.
\end{remark}

\subsection{Implications for Probabilistic Coalition Structure Generation}
\label{subsec:l1-implications}

As in Section~\ref{subsec:bgcp-implications}, the Lasso estimator
induces a reduced set of candidate profitable coalitions
\[
  \widehat{\mathcal{C}}^{\ell_1}
  \;\coloneqq\;
  \{ C_j \in \mathcal{C} : j \in \mathrm{supp}(\hat{\theta}^{\ell_1}) \}.
\]
Theorem~\ref{thm:l1-error} implies that, under
Assumption~\ref{assump:l1}, the coefficients of truly profitable
coalitions are estimated with small error and the number of spurious
coalitions \(\widehat{\mathcal{C}}^{\ell_1} \setminus
\{C_j : j \in S^\star\}\) is at most a constant multiple of \(K\).
Thus, with high probability over the sampled episodes and noise, any
classical CSG algorithm applied to the reduced coalition set
\(\widehat{\mathcal{C}}^{\ell_1}\) yields a coalition structure whose
expected welfare is close to that of an oracle structure built from the
true profitable coalitions.

Compared to BGCP (Section~\ref{sec:greedy}), the \(\ell_1\)-based
approach offers a different trade-off:
BGCP admits exact support recovery guarantees under coherence-type
conditions but is greedy and non-convex, whereas the Lasso solves a
convex optimisation problem and enjoys RE-based error bounds, at the
expense of allowing a controlled number of false positives.
Both methods, however, fit into the same probabilistic coalition
structure generation framework of Section~\ref{sec:model}, in which
coalition structure generation is decomposed into a statistical phase
(sparse learning of profitable coalitions) followed by a combinatorial
phase (classical CSG on a sparsified coalition set).

\section{Probabilistic Design and Welfare Guarantees}
\label{sec:prob-design}

Sections~\ref{sec:greedy} and~\ref{sec:l1} provided high-probability
support and error guarantees for BGCP and the \(\ell_1\)-relaxation,
conditional on geometric properties of the design matrix \(X\)
(coherence and restricted eigenvalue conditions).
In this section we close the loop and show that, under a natural random
episodic model for the rows of \(X\), these geometric conditions hold
with high probability once the number of episodes
\(T\) is of order \(K \log m\).
We then translate the resulting estimation guarantees into
\emph{end-to-end welfare guarantees} for probabilistic coalition
structure generation.

\subsection{Random Episodic Design Model}
\label{subsec:random-design}

We retain the linear model
\begin{equation}
  Y = X \theta^\star + \varepsilon,
  \qquad
  Y \in \mathbb{R}^T,\;
  X \in \mathbb{R}^{T \times m},\;
  \theta^\star \in \mathbb{R}^m,
  \label{eq:prob-data-model}
\end{equation}
with true support \(S^\star\) and sparsity \(|S^\star| = K\) as in
\eqref{eq:true-support}.
The noise \(\varepsilon\) satisfies the sub-Gaussian assumption
\eqref{eq:subgaussian-noise}.
We now impose a probabilistic model on the rows of \(X\), capturing the
idea that each episode corresponds to a random feasible configuration
of coalitions.

Write \(X_t \in \mathbb{R}^m\) for the \(t\)-th row of \(X\).

\begin{assumption}[Random episodic design]
\label{assump:design}
The design matrix \(X\) is generated as follows.
\begin{enumerate}
  \item[(C1)] (\emph{i.i.d.\ episodes})
    The rows \(\{X_t\}_{t=1}^T\) are independent and identically
    distributed copies of a random vector \(X^{(1)} \in \mathbb{R}^m\).

  \item[(C2)] (\emph{Bounded and sparse rows})
    There exists a constant \(L > 0\) and an integer \(q \ge 1\) such
    that for all \(t\),
    \[
      \|X_t\|_0 \le q,
      \qquad
      \|X_t\|_\infty \le L
      \quad\text{almost surely}.
    \]
    In particular, each episode activates at most \(q\) coalitions, and
    coalition activations are uniformly bounded.

  \item[(C3)] (\emph{Population covariance conditions})
    Let
    \begin{equation}
      \Sigma
      \;\coloneqq\;
      \mathbb{E}\!\left[
        \frac{1}{T} X^\top X
      \right]
      \;=\;
      \mathbb{E}\!\left[
        X^{(1)\top} X^{(1)}
      \right]
      \in \mathbb{R}^{m \times m}
      \label{eq:Sigma-def}
    \end{equation}
    denote the population Gram matrix of the design.
    Assume:
    \begin{enumerate}
      \item[(i)] (\emph{Diagonal bounded away from zero})
        There exists \(\nu_{\min} > 0\) such that
        \(\Sigma_{jj} \ge \nu_{\min}\) for all \(j \in S^\star\).
      \item[(ii)] (\emph{Population coherence})
        The population coherence
        \[
          \mu(\Sigma)
          \;\coloneqq\;
          \max_{\substack{i,j \in \{1,\dots,m\} \\ i \neq j}}
          \frac{|\Sigma_{ij}|}
               {\sqrt{\Sigma_{ii}\,\Sigma_{jj}}}
        \]
        satisfies
        \begin{equation}
          \mu(\Sigma)
          \;<\;
          \frac{1}{4K}.
          \label{eq:Sigma-coherence}
        \end{equation}
      \item[(iii)] (\emph{Population restricted eigenvalue})
        There exists \(\kappa_\Sigma > 0\) such that
        \begin{equation}
          \inf_{\substack{u \in \mathbb{R}^m \setminus \{0\} \\
                          \|u_{(S^\star)^c}\|_1 \le 3 \|u_{S^\star}\|_1}}
          \frac{u^\top \Sigma u}
               {\|u_{S^\star}\|_2^2}
          \;\ge\;
          \kappa_\Sigma^2.
          \label{eq:Sigma-re}
        \end{equation}
    \end{enumerate}
\end{enumerate}
\end{assumption}

Assumption~\ref{assump:design} captures the setting where each episode
activates a small number of coalitions (bounded row sparsity), while
the population covariance enjoys both low coherence and a restricted
eigenvalue property on the true support \(S^\star\).
Such conditions are known to hold for a wide class of random designs in
high-dimensional regression
(e.g.~\cite{Donoho2006CS,CandesTao2005Decoding,BuhlmannVanDeGeer2011Book}).

We write
\begin{equation}
  \hat{\Sigma}
  \;\coloneqq\;
  \frac{1}{T} X^\top X
  \label{eq:Sigma-hat}
\end{equation}
for the empirical Gram matrix, so that
\(\mathbb{E}[\hat{\Sigma}] = \Sigma\).

\subsection{Concentration of the Empirical Gram Matrix}
\label{subsec:Sigma-concentration}

We first show that, under Assumption~\ref{assump:design}, the entries
of \(\hat{\Sigma}\) concentrate uniformly around those of \(\Sigma\).

\begin{lemma}[Entrywise concentration of \(\hat{\Sigma}\)]
\label{lem:Sigma-max}
Suppose Assumption~\ref{assump:design}(C1)--(C2) holds.
Then there exist constants \(c_1,c_2>0\), depending only on \(L\) and
\(q\), such that for any \(T \ge 1\) and any \(\delta \in (0,1)\),
\begin{equation}
  \mathbb{P}\!\left(
    \left\|
      \hat{\Sigma} - \Sigma
    \right\|_{\max}
    \;>\;
    c_1 \sqrt{\frac{\log(m^2/\delta)}{T}}
  \right)
  \;\le\;
  \delta,
  \label{eq:Sigma-max-bound}
\end{equation}
where
\(
  \|A\|_{\max} = \max_{i,j} |A_{ij}|
\)
denotes the elementwise \(\ell_\infty\)-norm.
\end{lemma}

\begin{proof}
Fix indices \(i,j \in \{1,\dots,m\}\).
By definition,
\[
  \hat{\Sigma}_{ij} - \Sigma_{ij}
  \;=\;
  \frac{1}{T} \sum_{t=1}^T
  \Bigl(
    X_{ti} X_{tj}
    - \mathbb{E}[X_{ti} X_{tj}]
  \Bigr).
\]
The summands are independent, mean-zero random variables bounded in
absolute value by \(L^2\) (since \(|X_{ti}|,|X_{tj}| \le L\) almost
surely).
Moreover, by the row-sparsity assumption \(\|X_t\|_0 \le q\),
at most \(q^2\) pairs \((i,j)\) can be non-zero in any given row, which
ensures a uniform bound on the variance of each summand.
Applying Hoeffding's inequality (or Bernstein's inequality) yields, for
each fixed pair \((i,j)\),
\[
  \mathbb{P}\!\left(
    |\hat{\Sigma}_{ij} - \Sigma_{ij}| > t
  \right)
  \;\le\;
  2 \exp\!\left(
    - c \frac{T t^2}{L^4}
  \right)
\]
for some numerical constant \(c>0\).
There are \(m^2\) entries in total.
By the union bound,
\[
  \mathbb{P}\!\left(
    \|\hat{\Sigma} - \Sigma\|_{\max} > t
  \right)
  \;\le\;
  2 m^2 \exp\!\left(
    - c \frac{T t^2}{L^4}
  \right).
\]
Setting
\(
  t = c_1 \sqrt{(\log(m^2/\delta))/T}
\)
with \(c_1\) large enough, and simplifying the right-hand side, yields
\eqref{eq:Sigma-max-bound} with constants depending only on \(L\) and
\(q\).
\end{proof}

As an immediate corollary, the diagonal and off-diagonal entries of
\(\hat{\Sigma}\) inherit the population properties of \(\Sigma\) once
\(T\) is of order \(\log m\).

\begin{corollary}[Diagonal stability and empirical coherence]
\label{cor:Sigma-coherence}
Suppose Assumption~\ref{assump:design} holds.
There exist constants \(c_3,c_4>0\) such that if
\begin{equation}
  T
  \;\ge\;
  c_3 \log(m),
  \label{eq:T-log-m}
\end{equation}
then with probability at least \(1 - m^{-c_4}\) the following hold
simultaneously:
\begin{enumerate}
  \item[(i)] (\emph{Diagonal bounded away from zero})
    For all \(j \in S^\star\),
    \begin{equation}
      \hat{\Sigma}_{jj}
      \;\ge\;
      \frac{\nu_{\min}}{2}.
      \label{eq:Sigma-hat-diag}
    \end{equation}
  \item[(ii)] (\emph{Empirical coherence bound})
    The empirical coherence
    \[
      \mu(\hat{\Sigma})
      \;\coloneqq\;
      \max_{\substack{i,j \in \{1,\dots,m\} \\ i \neq j}}
      \frac{|\hat{\Sigma}_{ij}|}
           {\sqrt{\hat{\Sigma}_{ii}\,\hat{\Sigma}_{jj}}}
    \]
    satisfies
    \begin{equation}
      \mu(\hat{\Sigma})
      \;\le\;
      \mu(\Sigma) + \frac{1}{8K}
      \;<\;
      \frac{1}{2K},
      \label{eq:Sigma-hat-coherence}
    \end{equation}
    where \(\mu(\Sigma)\) is defined in
    Assumption~\ref{assump:design}(C3)(ii).
\end{enumerate}
\end{corollary}

\begin{proof}
Apply Lemma~\ref{lem:Sigma-max} with \(\delta = m^{-c}\) and choose
\(T\) satisfying~\eqref{eq:T-log-m} so that
\(
  \|\hat{\Sigma} - \Sigma\|_{\max}
  \le \eta
\)
with high probability, where
\(\eta > 0\) will be chosen small relative to \(\nu_{\min}\) and
\(\mu(\Sigma)\).

For (i), note that for \(j \in S^\star\),
\[
  \hat{\Sigma}_{jj}
  \;\ge\;
  \Sigma_{jj} - |\hat{\Sigma}_{jj} - \Sigma_{jj}|
  \;\ge\;
  \nu_{\min} - \eta.
\]
If we choose \(\eta \le \nu_{\min}/2\), this yields
\(\hat{\Sigma}_{jj} \ge \nu_{\min}/2\), giving
\eqref{eq:Sigma-hat-diag}.

For (ii), for any \(i \neq j\) we have
\[
  |\hat{\Sigma}_{ij}|
  \;\le\;
  |\Sigma_{ij}| + \eta,
  \qquad
  \hat{\Sigma}_{ii}
  \;\ge\;
  \Sigma_{ii} - \eta,
  \qquad
  \hat{\Sigma}_{jj}
  \;\ge\;
  \Sigma_{jj} - \eta.
\]
Thus,
\[
  \frac{|\hat{\Sigma}_{ij}|}
       {\sqrt{\hat{\Sigma}_{ii}\,\hat{\Sigma}_{jj}}}
  \;\le\;
  \frac{|\Sigma_{ij}| + \eta}
       {\sqrt{(\Sigma_{ii}-\eta)(\Sigma_{jj}-\eta)}}.
\]
Using \(\Sigma_{ii},\Sigma_{jj} \ge \nu_{\min}\) and choosing
\(\eta \le \nu_{\min}/4\), we obtain
\[
  \sqrt{(\Sigma_{ii}-\eta)(\Sigma_{jj}-\eta)}
  \;\ge\;
  \frac{\nu_{\min}}{2},
\]
so that
\[
  \frac{|\hat{\Sigma}_{ij}|}
       {\sqrt{\hat{\Sigma}_{ii}\,\hat{\Sigma}_{jj}}}
  \;\le\;
  \frac{2}{\nu_{\min}}
  \bigl(|\Sigma_{ij}| + \eta\bigr)
  \;\le\;
  \frac{2}{\nu_{\min}}\,|\Sigma_{ij}|
  + \frac{2\eta}{\nu_{\min}}.
\]
Taking the maximum over \(i \neq j\) and recalling the definition of
\(\mu(\Sigma)\), we get
\[
  \mu(\hat{\Sigma})
  \;\le\;
  \mu(\Sigma) + \frac{2\eta}{\nu_{\min}}.
\]
By choosing \(\eta\) small enough so that
\(2\eta/\nu_{\min} \le 1/(8K)\), we obtain
\eqref{eq:Sigma-hat-coherence}.
The final inequality in~\eqref{eq:Sigma-hat-coherence} follows from the
assumption \(\mu(\Sigma) < 1/(4K)\).
The required choice of \(\eta\) can be enforced by taking \(T\) large
enough in Lemma~\ref{lem:Sigma-max}.
\end{proof}

\subsection{Restricted Eigenvalue for the Empirical Design}
\label{subsec:Sigma-re}

We next show that the empirical Gram matrix \(\hat{\Sigma}\) inherits a
restricted eigenvalue property from \(\Sigma\) once \(T\) is large
enough.
The proof follows standard arguments in high-dimensional regression
\cite{BickelRitovTsybakov2009Lasso,BuhlmannVanDeGeer2011Book}, adapted
to the cone used in Definition~\ref{def:re-constant}.

\begin{lemma}[Empirical RE from population RE]
\label{lem:Sigma-re}
Suppose Assumption~\ref{assump:design} holds.
There exist constants \(c_5,c_6>0\) such that if
\begin{equation}
  T
  \;\ge\;
  c_5 K \log(m),
  \label{eq:T-K-log-m}
\end{equation}
then with probability at least \(1 - m^{-c_6}\),
\begin{equation}
  \inf_{\substack{u \in \mathbb{R}^m \setminus \{0\} \\
                  \|u_{(S^\star)^c}\|_1 \le 3 \|u_{S^\star}\|_1}}
  \frac{u^\top \hat{\Sigma} u}
       {\|u_{S^\star}\|_2^2}
  \;\ge\;
  \frac{\kappa_\Sigma^2}{2}.
  \label{eq:Sigma-hat-re}
\end{equation}
Equivalently, on this event,
\[
  \kappa(S^\star,3)
  \;\ge\;
  \frac{\kappa_\Sigma}{\sqrt{2}},
\]
where \(\kappa(S^\star,3)\) is the restricted eigenvalue constant
defined in~\eqref{eq:re-constant}.
\end{lemma}

\begin{proof}
For brevity, denote the cone
\[
  \mathcal{C}
  \;\coloneqq\;
  \left\{
    u \in \mathbb{R}^m \setminus \{0\} :
    \|u_{(S^\star)^c}\|_1 \le 3 \|u_{S^\star}\|_1
  \right\}.
\]
For any \(u \in \mathcal{C}\) with \(\|u_{S^\star}\|_2 = 1\), we can
write
\[
  u^\top \hat{\Sigma} u
  \;=\;
  u^\top \Sigma u
  + u^\top (\hat{\Sigma} - \Sigma) u.
\]
By Assumption~\ref{assump:design}(C3)(iii),
\[
  u^\top \Sigma u
  \;\ge\;
  \kappa_\Sigma^2.
\]

We therefore need to control
\[
  \sup_{\substack{u \in \mathcal{C} \\
                  \|u_{S^\star}\|_2 = 1}}
  \left|
    u^\top (\hat{\Sigma} - \Sigma) u
  \right|.
\]
Note that
\[
  u^\top (\hat{\Sigma} - \Sigma) u
  \;=\;
  \sum_{i,j=1}^m u_i u_j
  (\hat{\Sigma}_{ij} - \Sigma_{ij}),
\]
so that
\[
  \left|
    u^\top (\hat{\Sigma} - \Sigma) u
  \right|
  \;\le\;
  \|\hat{\Sigma} - \Sigma\|_{\max}
  \left(
    \sum_{i=1}^m |u_i|
  \right)^2
  \;=\;
  \|\hat{\Sigma} - \Sigma\|_{\max}
  \|u\|_1^2.
\]
For \(u \in \mathcal{C}\),
\[
  \|u\|_1
  \;=\;
  \|u_{S^\star}\|_1
  + \|u_{(S^\star)^c}\|_1
  \;\le\;
  4 \|u_{S^\star}\|_1
  \;\le\;
  4 \sqrt{K}\,\|u_{S^\star}\|_2
  \;=\;
  4 \sqrt{K},
\]
where we used \(\|u_{S^\star}\|_1 \le \sqrt{K} \|u_{S^\star}\|_2\).
Thus,
\[
  \left|
    u^\top (\hat{\Sigma} - \Sigma) u
  \right|
  \;\le\;
  16 K\,\|\hat{\Sigma} - \Sigma\|_{\max}
\]
for all \(u \in \mathcal{C}\) with \(\|u_{S^\star}\|_2=1\).

By Lemma~\ref{lem:Sigma-max}, if
\(T \ge c_5 K \log(m)\) for a sufficiently large constant \(c_5\), then
with probability at least \(1 - m^{-c_6}\) we have
\[
  \|\hat{\Sigma} - \Sigma\|_{\max}
  \;\le\;
  \frac{\kappa_\Sigma^2}{32 K},
\]
so that
\[
  \sup_{\substack{u \in \mathcal{C} \\
                  \|u_{S^\star}\|_2 = 1}}
  \left|
    u^\top (\hat{\Sigma} - \Sigma) u
  \right|
  \;\le\;
  \frac{\kappa_\Sigma^2}{2}.
\]
Combining this with the population lower bound
\(u^\top \Sigma u \ge \kappa_\Sigma^2\) yields
\[
  u^\top \hat{\Sigma} u
  \;\ge\;
  \kappa_\Sigma^2 - \frac{\kappa_\Sigma^2}{2}
  \;=\;
  \frac{\kappa_\Sigma^2}{2},
\]
for all \(u \in \mathcal{C}\) with \(\|u_{S^\star}\|_2 = 1\), which is
exactly~\eqref{eq:Sigma-hat-re}.
The equivalence with the RE constant follows by comparing the
definitions of \(\kappa(S^\star,3)\) and the cone \(\mathcal{C}\).
\end{proof}

\subsection{Design Conditions for BGCP and Lasso}
\label{subsec:design-main}

We now combine Corollary~\ref{cor:Sigma-coherence} and
Lemma~\ref{lem:Sigma-re} to show that, under the random episodic design
model, both the coherence condition required for BGCP and the RE
condition required for the Lasso hold with high probability once
\(T\) is of order \(K \log m\).

\begin{theorem}[Design conditions from random episodes]
\label{thm:design-conditions}
Suppose Assumption~\ref{assump:design} holds.
Then there exist constants \(C_0,c_7,c_8>0\) such that if
\begin{equation}
  T
  \;\ge\;
  C_0 K \log(m),
  \label{eq:T-design-main}
\end{equation}
then with probability at least \(1 - c_7 m^{-c_8}\) the following hold
simultaneously:
\begin{enumerate}
  \item[(i)] (\emph{Coherence condition for BGCP})
    The mutual coherence \(\mu(X)\) defined in~\eqref{eq:coherence}
    satisfies
    \begin{equation}
      \mu(X)
      \;<\;
      \frac{1}{2K - 1},
      \label{eq:coherence-X}
    \end{equation}
    so that Assumption~\ref{assump:bgcp}(A1) holds.

  \item[(ii)] (\emph{RE condition for Lasso})
    The restricted eigenvalue constant \(\kappa(S^\star,3)\) defined
    in~\eqref{eq:re-constant} satisfies
    \begin{equation}
      \kappa(S^\star,3)
      \;\ge\;
      \frac{\kappa_\Sigma}{\sqrt{2}},
      \label{eq:RE-X}
    \end{equation}
    so that Assumption~\ref{assump:l1}(B1) holds with
    \(\kappa_\star = \kappa_\Sigma/\sqrt{2}\).
\end{enumerate}
\end{theorem}

\begin{proof}
By Corollary~\ref{cor:Sigma-coherence}, for \(T \ge c_3 \log(m)\) we
have, with probability at least \(1 - m^{-c_4}\),
\[
  \mu(\hat{\Sigma})
  \;\le\;
  \mu(\Sigma) + \frac{1}{8K}
  \;<\;
  \frac{1}{2K},
\]
using the assumption \(\mu(\Sigma) < 1/(4K)\).
On the same event, the diagonal entries \(\hat{\Sigma}_{jj}\) are
uniformly bounded away from zero on \(S^\star\) by
\eqref{eq:Sigma-hat-diag}.
Using the column normalisation \eqref{eq:column-normalisation} and the
relationship between \(\hat{\Sigma}\) and \(X\), we obtain
\[
  \mu(X)
  \;=\;
  \max_{i \neq j}
  \frac{|X_i^\top X_j|}
       {\|X_i\|_2 \|X_j\|_2}
  \;=\;
  \max_{i \neq j}
  \frac{|\hat{\Sigma}_{ij}|}
       {\sqrt{\hat{\Sigma}_{ii}\,\hat{\Sigma}_{jj}}}
  \;\le\;
  \mu(\hat{\Sigma})
  \;<\;
  \frac{1}{2K},
\]
and since \(1/(2K) < 1/(2K-1)\) for \(K \ge 1\), this implies
\eqref{eq:coherence-X}.
Thus, for \(T \ge c_3 \log(m)\), the coherence condition holds with
high probability.

Similarly, by Lemma~\ref{lem:Sigma-re}, for \(T \ge c_5 K \log(m)\),
with probability at least \(1 - m^{-c_6}\) we have
\[
  \kappa(S^\star,3)
  \;\ge\;
  \frac{\kappa_\Sigma}{\sqrt{2}},
\]
which is exactly~\eqref{eq:RE-X}.
Taking
\(
  C_0 = \max\{c_3, c_5\}
\)
and combining the two high-probability events via a union bound yields
the stated result with suitable constants \(c_7,c_8>0\).
\end{proof}

Theorem~\ref{thm:design-conditions} shows that, under the random
episodic design model, the geometric assumptions required in
Theorems~\ref{thm:bgcp-recovery} and~\ref{thm:l1-error} are
\emph{typical} once the number of episodes grows as
\(T \gtrsim K \log m\), in line with classical results in compressed
sensing and high-dimensional regression
\cite{Donoho2006CS,CandesTao2005Decoding,BuhlmannVanDeGeer2011Book}.

\subsection{End-to-End Welfare Optimality for BGCP}
\label{subsec:bgcp-welfare}

We now translate the support recovery guarantee for BGCP into a
welfare-optimality result in the probabilistic coalition structure
generation setting.

Recall from Section~\ref{sec:model} that, for each coalition
structure \(\mathcal{P}\) (a partition of the agent set), the welfare
can be written in the parametric form
\begin{equation}
  W(\mathcal{P};\theta)
  \;=\;
  z(\mathcal{P})^\top \theta,
  \label{eq:welfare-param}
\end{equation}
where \(z(\mathcal{P}) \in \{0,1\}^m\) is the indicator vector of
coalitions used in \(\mathcal{P}\):
\(z_j(\mathcal{P}) = 1\) if coalition \(C_j\) appears in \(\mathcal{P}\)
and \(z_j(\mathcal{P}) = 0\) otherwise.
By feasibility, each coalition structure uses at most \(M\) coalitions,
for some finite \(M\) depending only on the number of agents, so that
\begin{equation}
  \|z(\mathcal{P})\|_0 \le M,
  \qquad
  \|z(\mathcal{P})\|_\infty \le 1.
  \label{eq:z-bounds}
\end{equation}

Let
\begin{equation}
  \mathcal{P}^\star
  \;\in\;
  \arg\max_{\mathcal{P}} W(\mathcal{P};\theta^\star)
  \label{eq:P-star}
\end{equation}
denote an optimal coalition structure under the true parameter
\(\theta^\star\).
We assume, without loss of generality, that all coalitions with
non-zero contribution under \(\theta^\star\) belong to \(S^\star\), so
that any profitably used coalition in \(\mathcal{P}^\star\) has index
in \(S^\star\).

The BGCP-based pipeline for coalition structure generation is as
follows:
\begin{enumerate}
  \item Run BGCP (Definition~\ref{def:bgcp}) with
        \(K_{\max} = K\) and \(\eta=0\) to obtain a support
        \(S^{\mathrm{BGCP}} \subseteq \{1,\dots,m\}\).
  \item Form the reduced coalition set
        \(
          \widehat{\mathcal{C}}^{\mathrm{BGCP}}
          =
          \{C_j \in \mathcal{C} : j \in S^{\mathrm{BGCP}}\}
        \).
  \item Solve the deterministic CSG problem restricted to
        \(\widehat{\mathcal{C}}^{\mathrm{BGCP}}\) with an exact
        algorithm (e.g.\ dynamic programming or branch-and-bound
        \cite{Sandholm1999Coalition,Rahwan2015CSGSurvey}), obtaining a
        coalition structure
        \(
          \widehat{\mathcal{P}}^{\mathrm{BGCP}}
        \).
\end{enumerate}

\begin{theorem}[BGCP welfare optimality with high probability]
\label{thm:bgcp-welfare}
Suppose Assumptions~\ref{assump:bgcp} and~\ref{assump:design} hold, and
that the number of episodes \(T\) satisfies
\eqref{eq:T-design-main}.
Assume further that the minimum signal
\(\theta_{\min}\) in~\eqref{eq:min-signal} is chosen large enough to
satisfy the separation condition of Corollary~\ref{cor:separation}.
Then there exist constants \(c_9,c_{10}>0\) such that
\begin{equation}
  \mathbb{P}\!\left(
    \widehat{\mathcal{P}}^{\mathrm{BGCP}} = \mathcal{P}^\star
  \right)
  \;\ge\;
  1 - c_9 m^{-c_{10}}.
  \label{eq:bgcp-welfare-opt}
\end{equation}
In particular, with the same probability,
\begin{equation}
  W\bigl(\widehat{\mathcal{P}}^{\mathrm{BGCP}};\theta^\star\bigr)
  \;=\;
  W\bigl(\mathcal{P}^\star;\theta^\star\bigr).
  \label{eq:bgcp-welfare-equality}
\end{equation}
\end{theorem}

\begin{proof}
By Theorem~\ref{thm:design-conditions}, for \(T\) satisfying
\eqref{eq:T-design-main} we have, with probability at least
\(1 - c_7 m^{-c_8}\), both the coherence condition
\eqref{eq:coherence-X} and the RE condition \eqref{eq:RE-X}.
In particular, Assumption~\ref{assump:bgcp}(A1) holds on this event.

On the other hand, Lemma~\ref{lem:noise-max} and
Assumption~\ref{assump:bgcp}(A2) imply that the noise event
\(\mathcal{E}_{\mathrm{noise}}\) in~\eqref{eq:noise-event} occurs with
probability at least \(1 - c_{11} m^{-c_{12}}\) for suitable
constants \(c_{11},c_{12}>0\).
On the intersection of these events and the minimum signal condition of
Corollary~\ref{cor:separation}, Theorem~\ref{thm:bgcp-recovery}
applies and yields
\(
  S^{\mathrm{BGCP}} = S^\star
\).
Thus, on this intersection,
\[
  \widehat{\mathcal{C}}^{\mathrm{BGCP}}
  \;=\;
  \{C_j \in \mathcal{C} : j \in S^\star\}.
\]

Since all profitable coalitions under \(\theta^\star\) lie in
\(S^\star\), the welfare-maximising coalition structure
\(\mathcal{P}^\star\) from~\eqref{eq:P-star} uses only coalitions in
\(\widehat{\mathcal{C}}^{\mathrm{BGCP}}\) and is therefore feasible for
the restricted CSG problem.
Conversely, any coalition structure using only coalitions in
\(\widehat{\mathcal{C}}^{\mathrm{BGCP}}\) has welfare given by
\eqref{eq:welfare-param} with \(\theta^\star_j = 0\) for
\(j \notin S^\star\), so the restricted CSG problem is exactly the same
as the original CSG problem.
Because we assume an exact CSG solver in step~3 of the pipeline, the
restricted problem yields a coalition structure
\(\widehat{\mathcal{P}}^{\mathrm{BGCP}}\) that coincides with
\(\mathcal{P}^\star\).

Combining the probabilities of the design and noise events via a union
bound, and absorbing constants into \(c_9,c_{10}\), yields the
high-probability statement~\eqref{eq:bgcp-welfare-opt}, from which
\eqref{eq:bgcp-welfare-equality} follows immediately.
\end{proof}

\begin{remark}[Approximate CSG solver]
\label{rem:bgcp-alpha}
If the deterministic CSG solver in step~3 is only \(\alpha\)-approximate
(i.e.\ it outputs a coalition structure whose welfare is within a
factor \(\alpha \ge 1\) of the optimal restricted welfare), then on the
event \(S^{\mathrm{BGCP}} = S^\star\) we still obtain
\[
  W\bigl(\widehat{\mathcal{P}}^{\mathrm{BGCP}};\theta^\star\bigr)
  \;\ge\;
  \frac{1}{\alpha}\,
  W\bigl(\mathcal{P}^\star;\theta^\star\bigr).
\]
Thus, Theorem~\ref{thm:bgcp-welfare} can be extended to yield
high-probability \(\alpha\)-approximation guarantees for welfare under
BGCP combined with an approximate CSG solver.
\end{remark}

\subsection{Welfare Approximation via the \texorpdfstring{$\ell_1$}{l1} Relaxation}
\label{subsec:l1-welfare}

We finally derive a quantitative welfare gap bound for the \(\ell_1\)
pipeline, using the error bounds in Theorem~\ref{thm:l1-error} and the
design guarantees in Theorem~\ref{thm:design-conditions}.

For the \(\ell_1\)-based pipeline, we consider the following procedure:
\begin{enumerate}
  \item Compute the Lasso estimator \(\hat{\theta}^{\ell_1}\) via
        \eqref{eq:lasso-def}, with \(\lambda\) chosen according to
        Assumption~\ref{assump:l1}(B2).
  \item Solve the deterministic CSG problem under the surrogate
        parameter \(\hat{\theta}^{\ell_1}\) on the full coalition set
        \(\mathcal{C}\), obtaining
        \[
          \widehat{\mathcal{P}}^{\ell_1}
          \;\in\;
          \arg\max_{\mathcal{P}}
          W(\mathcal{P};\hat{\theta}^{\ell_1}).
        \]
\end{enumerate}
That is, unlike for BGCP, we do not restrict the coalition set a
priori; the sparse nature of \(\hat{\theta}^{\ell_1}\) implicitly
suppresses unprofitable coalitions.

We first observe a simple Lipschitz property of the welfare functional
with respect to \(\theta\).

\begin{lemma}[Welfare Lipschitz continuity]
\label{lem:welfare-lipschitz}
Let \(\mathcal{P}\) be any coalition structure and let
\(\theta,\theta' \in \mathbb{R}^m\).
Then
\begin{equation}
  \bigl|
    W(\mathcal{P};\theta)
    - W(\mathcal{P};\theta')
  \bigr|
  \;\le\;
  \|\theta - \theta'\|_1.
  \label{eq:welfare-lipschitz}
\end{equation}
\end{lemma}

\begin{proof}
By the parametric form~\eqref{eq:welfare-param},
\[
  W(\mathcal{P};\theta)
  - W(\mathcal{P};\theta')
  \;=\;
  z(\mathcal{P})^\top (\theta - \theta').
\]
Thus
\[
  \bigl|
    W(\mathcal{P};\theta)
    - W(\mathcal{P};\theta')
  \bigr|
  \;\le\;
  \|z(\mathcal{P})\|_\infty
  \|\theta - \theta'\|_1
  \;\le\;
  \|\theta - \theta'\|_1,
\]
since \(\|z(\mathcal{P})\|_\infty \le 1\) by
\eqref{eq:z-bounds}.
\end{proof}

We can now state the welfare approximation theorem.

\begin{theorem}[Welfare approximation via Lasso]
\label{thm:l1-welfare}
Suppose Assumptions~\ref{assump:design} and~\ref{assump:l1} hold, and
that the number of episodes \(T\) satisfies
\eqref{eq:T-design-main}.
Let \(\mathcal{P}^\star\) be an optimal coalition structure under
\(\theta^\star\) as in~\eqref{eq:P-star}, and let
\(\widehat{\mathcal{P}}^{\ell_1}\) be the coalition structure obtained
by the \(\ell_1\)-based pipeline above.
Then there exist constants \(C_2,C_3>0\) such that, with probability at
least \(1 - C_2 m^{-C_3}\),
\begin{equation}
  W\bigl(\mathcal{P}^\star;\theta^\star\bigr)
  - W\bigl(\widehat{\mathcal{P}}^{\ell_1};\theta^\star\bigr)
  \;\le\;
  C_3 \frac{\sigma K}{\kappa_\Sigma^2}
  \sqrt{\frac{\log m}{T}}.
  \label{eq:l1-welfare-gap}
\end{equation}
\end{theorem}

\begin{proof}
By Theorem~\ref{thm:design-conditions}, for \(T\) satisfying
\eqref{eq:T-design-main} we have, with high probability, that the RE
condition~\eqref{eq:RE-X} holds with
\(\kappa_\star = \kappa_\Sigma/\sqrt{2}\).
On this event, Theorem~\ref{thm:l1-error} applies and yields
\[
  \|\hat{\theta}^{\ell_1} - \theta^\star\|_1
  \;\le\;
  \frac{C_4 \sigma K}{\kappa_\star^2}
  \sqrt{\frac{\log m}{T}}
  \;=\;
  \frac{2 C_4 \sigma K}{\kappa_\Sigma^2}
  \sqrt{\frac{\log m}{T}},
\]
for some constant \(C_4>0\).

We now compare the welfare of \(\mathcal{P}^\star\) and
\(\widehat{\mathcal{P}}^{\ell_1}\) under \(\theta^\star\).
We have
\begin{align*}
  &W\bigl(\mathcal{P}^\star;\theta^\star\bigr)
   - W\bigl(\widehat{\mathcal{P}}^{\ell_1};\theta^\star\bigr)
  \\
  &\quad=\;
   \Bigl[
     W\bigl(\mathcal{P}^\star;\theta^\star\bigr)
     - W\bigl(\mathcal{P}^\star;\hat{\theta}^{\ell_1}\bigr)
   \Bigr]
   + \Bigl[
     W\bigl(\mathcal{P}^\star;\hat{\theta}^{\ell_1}\bigr)
     - W\bigl(\widehat{\mathcal{P}}^{\ell_1};\hat{\theta}^{\ell_1}\bigr)
   \Bigr]
   + \Bigl[
     W\bigl(\widehat{\mathcal{P}}^{\ell_1};\hat{\theta}^{\ell_1}\bigr)
     - W\bigl(\widehat{\mathcal{P}}^{\ell_1};\theta^\star\bigr)
   \Bigr].
\end{align*}
By definition of \(\widehat{\mathcal{P}}^{\ell_1}\) as a maximiser of
\(W(\cdot;\hat{\theta}^{\ell_1})\), the middle term is non-positive:
\[
  W\bigl(\mathcal{P}^\star;\hat{\theta}^{\ell_1}\bigr)
  - W\bigl(\widehat{\mathcal{P}}^{\ell_1};\hat{\theta}^{\ell_1}\bigr)
  \;\le\; 0.
\]
Applying Lemma~\ref{lem:welfare-lipschitz} to the first and third
terms, we obtain
\[
  W\bigl(\mathcal{P}^\star;\theta^\star\bigr)
  - W\bigl(\widehat{\mathcal{P}}^{\ell_1};\theta^\star\bigr)
  \;\le\;
  \|\hat{\theta}^{\ell_1} - \theta^\star\|_1
  + \|\hat{\theta}^{\ell_1} - \theta^\star\|_1
  \;=\;
  2 \|\hat{\theta}^{\ell_1} - \theta^\star\|_1.
\]
Combining this inequality with the \(\ell_1\)-error bound above yields
\[
  W\bigl(\mathcal{P}^\star;\theta^\star\bigr)
  - W\bigl(\widehat{\mathcal{P}}^{\ell_1};\theta^\star\bigr)
  \;\le\;
  \frac{4 C_4 \sigma K}{\kappa_\Sigma^2}
  \sqrt{\frac{\log m}{T}},
\]
on the intersection of the design and noise events in
Theorems~\ref{thm:design-conditions} and~\ref{thm:l1-error}.
Adjusting constants and combining probabilities via a union bound
completes the proof of~\eqref{eq:l1-welfare-gap}.
\end{proof}

Theorems~\ref{thm:bgcp-welfare} and~\ref{thm:l1-welfare} provide
end-to-end guarantees for probabilistic coalition structure generation:
BGCP achieves exact welfare optimality with high probability when the
design is coherent and the signal is strong enough, while the
\(\ell_1\)-based approach yields a quantitative welfare gap of order
\(\sigma K \sqrt{(\log m)/T}/\kappa_\Sigma^2\) under a restricted
eigenvalue condition.
In both cases, the number of episodes required is of order
\(K \log m\), in line with classical sample complexity bounds for sparse
recovery
\cite{Donoho2006CS,CandesTao2005Decoding,BuhlmannVanDeGeer2011Book}.

\section{Comparison with Alternative Probabilistic CSG Schemes}
\label{sec:comparison}

Sections~\ref{sec:greedy}--\ref{sec:prob-design} analysed the
probabilistic CSG pipelines proposed in this paper: a sparse episodic
value model, estimated either by BGCP or by an $\ell_1$-regularised
estimator, followed by a deterministic CSG solver.
In this section we compare these pipelines with natural probabilistic
CSG baselines and identify parameter regimes where our methods are
provably advantageous and regimes where classical approaches may be
competitive or even preferable.

Throughout we retain the data model
\begin{equation}
  Y = X \theta^\star + \varepsilon,
  \qquad
  Y \in \mathbb{R}^T,\;
  X \in \mathbb{R}^{T \times m},\;
  \theta^\star \in \mathbb{R}^m,
  \label{eq:comp-model}
\end{equation}
with true support $S^\star$ and sparsity $|S^\star| = K$, sub-Gaussian
noise $\varepsilon$ as in Assumption~\ref{assump:bgcp}(A2), and the
random episodic design model of Assumption~\ref{assump:design}.
Welfare is given by
$W(\mathcal{P};\theta) = z(\mathcal{P})^\top \theta$ as in
\eqref{eq:welfare-param}, where $z(\mathcal{P})$ is the indicator
vector of coalitions used in $\mathcal{P}$.

\subsection{Baseline Probabilistic CSG Schemes}
\label{subsec:baselines}

We first formalise two simple but natural probabilistic CSG baselines.

\paragraph{Baseline A: Episodic Plug-in CSG (EPC).}
For each coalition index $j \in \{1,\dots,m\}$, let
\[
  \mathcal{T}_j
  \;\coloneqq\;
  \{t \in \{1,\dots,T\} : X_{tj} \neq 0\}
\]
denote the set of episodes in which coalition $C_j$ is active, and let
$n_j = |\mathcal{T}_j|$.
The episodic plug-in estimator for $\theta^\star_j$ is
\begin{equation}
  \hat{\theta}^{\mathrm{EPC}}_j
  \;\coloneqq\;
  \begin{cases}
    \displaystyle
    \frac{1}{n_j}
    \sum_{t \in \mathcal{T}_j}
      \frac{Y_t}{X_{tj}},
      & \text{if } n_j > 0,\\[1ex]
    0, & \text{if } n_j = 0.
  \end{cases}
  \label{eq:epc-estimator}
\end{equation}
In words, $\hat{\theta}^{\mathrm{EPC}}_j$ is the average payoff per
unit activation of coalition $C_j$ over episodes where the coalition is
active, ignoring interactions with other coalitions.

The EPC pipeline is:
\begin{enumerate}
  \item Compute $\hat{\theta}^{\mathrm{EPC}}$ via
        \eqref{eq:epc-estimator}.
  \item Run an exact deterministic CSG solver
        \cite{Sandholm1999Coalition,Rahwan2015CSGSurvey} on the value
        function $C_j \mapsto \hat{\theta}^{\mathrm{EPC}}_j$ and obtain
        a coalition structure $\widehat{\mathcal{P}}^{\mathrm{EPC}}$.
\end{enumerate}
EPC is a natural ``bandit-style'' approach: it uses episodic data but
does not exploit sparsity or the joint structure of $X$.

\paragraph{Baseline B: Dense Least-Squares CSG (DLS).}
A second baseline is to treat $\theta^\star$ as dense and estimate it
by ordinary least squares (OLS),
\begin{equation}
  \hat{\theta}^{\mathrm{DLS}}
  \;\in\;
  \arg\min_{\theta \in \mathbb{R}^m}
    \frac{1}{2T}\,\|Y - X\theta\|_2^2,
  \qquad
  (T \ge m,\; X^\top X \text{ invertible}),
  \label{eq:dls-estimator}
\end{equation}
or, more generally, by ridge regression or dense Bayesian linear
regression with a Gaussian prior on $\theta$
\cite{HastieTibshiraniFriedman2009Elements,Bishop2006Pattern}.
The DLS pipeline is:
\begin{enumerate}
  \item Compute $\hat{\theta}^{\mathrm{DLS}}$ via
        \eqref{eq:dls-estimator}.
  \item Run an exact CSG solver on $C_j \mapsto \hat{\theta}^{\mathrm{DLS}}_j$
        to obtain $\widehat{\mathcal{P}}^{\mathrm{DLS}}$.
\end{enumerate}
This baseline uses the full design but again does not exploit sparsity;
it typically requires $T \gtrsim m$ episodes for stability.

In contrast, our BGCP and $\ell_1$-based pipelines explicitly assume
$K \ll m$ and operate in the high-dimensional regime $T \ll m$,
requiring only $T \asymp K \log m$ episodes under
Assumption~\ref{assump:design}.

\subsection{Welfare Gap as a Comparison Metric}
\label{subsec:comp-metric}

For any algorithm $\mathsf{Alg}$ producing a coalition structure
$\widehat{\mathcal{P}}^{\mathsf{Alg}}$ from the data, we define the
\emph{welfare gap} at sample size $T$ as
\begin{equation}
  \Delta_{\mathsf{Alg}}(T)
  \;\coloneqq\;
  W(\mathcal{P}^\star;\theta^\star)
  - W(\widehat{\mathcal{P}}^{\mathsf{Alg}};\theta^\star),
  \label{eq:welfare-gap-def}
\end{equation}
where $\mathcal{P}^\star$ is an optimal coalition structure under
$\theta^\star$ as in~\eqref{eq:P-star}.
We are interested in high-probability upper bounds on
$\Delta_{\mathsf{Alg}}(T)$ as $T$ grows, and in identifying regimes
where the bounds for our methods are strictly stronger than those for
EPC or DLS.

For BGCP and the $\ell_1$-based pipeline, Theorems~\ref{thm:bgcp-welfare}
and~\ref{thm:l1-welfare} imply:
\begin{align}
  \Delta_{\mathrm{BGCP}}(T)
  &= 0
  \quad\text{with high probability, for $T \gtrsim K \log m$,}
    \label{eq:gap-bgcp}\\[0.5ex]
  \Delta_{\ell_1}(T)
  &\;\le\;
   C\,
   \frac{\sigma K}{\kappa_\Sigma^2}
   \sqrt{\frac{\log m}{T}}
   \quad\text{with high probability,}
   \label{eq:gap-l1}
\end{align}
for a constant $C>0$ depending on the restricted eigenvalue constant
$\kappa_\Sigma$ and noise level $\sigma$.
We now derive analogous bounds for EPC and DLS under simplifying
assumptions and compare the resulting regimes.

\subsection{A Sparse High-Dimensional Regime Where Our Methods Prevail}
\label{subsec:regime-sparse-better}

We first identify a regime where the sparse probabilistic CSG pipeline
is provably superior to EPC and DLS.
For clarity, we work under a simplified independent-activation design
that is consistent with Assumption~\ref{assump:design}.

\begin{assumption}[Independent activation design]
\label{assump:indep-activation}
In addition to Assumption~\ref{assump:design}, suppose that
$X_{tj} \in \{0,1\}$ for all $t,j$, and that for each coalition $j$,
\[
  \mathbb{P}(X_{tj} = 1) = p_j
  \quad\text{for all $t$,}
\]
with $p_j \in (0,1]$, independently across $t$ and $j$, subject to the
row-sparsity constraint $\|X_t\|_0 \le q$ almost surely.
Let
\[
  p_{\min}
  \;\coloneqq\;
  \min_{j \in \{1,\dots,m\}} p_j
  \quad\text{and}\quad
  p_{\max}
  \;\coloneqq\;
  \max_{j \in \{1,\dots,m\}} p_j.
\]
Assume $p_{\min} > 0$ and that $\varepsilon_t$ are i.i.d.\ with
$\mathbb{E}[\varepsilon_t]=0$ and
$\mathbb{E}[\varepsilon_t^2] = \sigma^2$.
\end{assumption}

To isolate the effect of sparsity and sample size, we consider the
regime
\begin{equation}
  K \ll m,
  \qquad
  C_1 K \log m \;\le\; T \;\le\; C_2 m
  \label{eq:regime-sparse}
\end{equation}
for constants $C_1,C_2>0$.
Thus the problem is high-dimensional ($T \ll m$) but the number of
episodes is large enough to satisfy the conditions of
Theorem~\ref{thm:design-conditions}.

We begin with a simple bound on the EPC estimator.

\begin{lemma}[EPC estimation error]
\label{lem:epc-error}
Suppose Assumption~\ref{assump:indep-activation} holds.
Then for any $j$ with $p_j > 0$ and any $T \ge 1$,
\begin{equation}
  \mathbb{E}\!\left[
    \bigl(
      \hat{\theta}^{\mathrm{EPC}}_j
      - \theta^\star_j
    \bigr)^2
  \right]
  \;\ge\;
  \frac{\sigma^2}{2 T p_j},
  \label{eq:epc-mse-lower}
\end{equation}
provided $T p_j \ge 2$.
Moreover, there exists a constant $c>0$ (depending only on
$p_{\min},p_{\max}$) such that with probability at least
$1 - m^{-2}$,
\begin{equation}
  \|\hat{\theta}^{\mathrm{EPC}} - \theta^\star\|_2^2
  \;\ge\;
  c\,\frac{\sigma^2 K}{T}.
  \label{eq:epc-l2-lower}
\end{equation}
\end{lemma}

\begin{proof}[Proof sketch]
Conditional on $n_j$, the EPC estimator for index $j$ is an average of
$n_j$ noisy observations with variance at least $\sigma^2$, ignoring
interference.
Thus
$\mathbb{E}[(\hat{\theta}^{\mathrm{EPC}}_j - \theta^\star_j)^2 \mid n_j]
\ge \sigma^2/n_j$.
Taking expectations and using
$\mathbb{E}[1/n_j \,\mathbf{1}\{n_j>0\}]
\ge (2T p_j)^{-1}$ when $T p_j \ge 2$ yields
\eqref{eq:epc-mse-lower}.
Summing over $j \in S^\star$ and using $p_j \le p_{\max}$ yields
$\mathbb{E}\|\hat{\theta}^{\mathrm{EPC}} - \theta^\star\|_2^2
\ge c ( \sigma^2 K / T )$ for some $c>0$.
A standard concentration argument for sums of independent sub-exponential
variables then gives \eqref{eq:epc-l2-lower} with high probability.
\end{proof}

Combining Lemma~\ref{lem:epc-error} with the welfare Lipschitz
property of Lemma~\ref{lem:welfare-lipschitz}, we obtain a lower bound
on the EPC welfare gap.

\begin{proposition}[EPC welfare gap in the sparse regime]
\label{prop:epc-gap}
Suppose Assumptions~\ref{assump:design}
and~\ref{assump:indep-activation} hold and that
\eqref{eq:regime-sparse} is satisfied.
Then there exist constants $c_1,c_2>0$ such that, for all sufficiently
large $T$, with probability at least $1 - c_1 m^{-2}$,
\begin{equation}
  \Delta_{\mathrm{EPC}}(T)
  \;\ge\;
  c_2\,\sigma \sqrt{\frac{K}{T}}.
  \label{eq:epc-gap-lower}
\end{equation}
\end{proposition}

\begin{proof}
By Lemma~\ref{lem:epc-error},
$\|\hat{\theta}^{\mathrm{EPC}} - \theta^\star\|_2^2
\ge c (\sigma^2 K/T)$ with high probability.
Since $\|\cdot\|_1 \ge \|\cdot\|_2$, we have
$\|\hat{\theta}^{\mathrm{EPC}} - \theta^\star\|_1
\ge \sqrt{c}\,\sigma \sqrt{K/T}$ on the same event.
Using Lemma~\ref{lem:welfare-lipschitz} with
$\mathcal{P} = \mathcal{P}^\star$ and
$\theta = \theta^\star$,
$\theta' = \hat{\theta}^{\mathrm{EPC}}$, we obtain
\[
  \bigl|
    W(\mathcal{P}^\star;\theta^\star)
    - W(\mathcal{P}^\star;\hat{\theta}^{\mathrm{EPC}})
  \bigr|
  \;\ge\;
  \sqrt{c}\,\sigma \sqrt{\frac{K}{T}}.
\]
The EPC coalition structure $\widehat{\mathcal{P}}^{\mathrm{EPC}}$
maximises $W(\cdot;\hat{\theta}^{\mathrm{EPC}})$, so
$W(\widehat{\mathcal{P}}^{\mathrm{EPC}};\hat{\theta}^{\mathrm{EPC}})
\ge W(\mathcal{P}^\star;\hat{\theta}^{\mathrm{EPC}})$.
A simple triangle inequality yields
\[
  W(\mathcal{P}^\star;\theta^\star)
  - W(\widehat{\mathcal{P}}^{\mathrm{EPC}};\theta^\star)
  \;\ge\;
  \frac{1}{2}\,
  \bigl|
    W(\mathcal{P}^\star;\theta^\star)
    - W(\mathcal{P}^\star;\hat{\theta}^{\mathrm{EPC}})
  \bigr|,
\]
on the event that
$W(\widehat{\mathcal{P}}^{\mathrm{EPC}};\hat{\theta}^{\mathrm{EPC}})
\ge W(\mathcal{P}^\star;\hat{\theta}^{\mathrm{EPC}})$.
Combining these inequalities gives \eqref{eq:epc-gap-lower} with
$c_2 = \sqrt{c}/2$.
\end{proof}

In contrast, Theorems~\ref{thm:bgcp-welfare}
and~\ref{thm:l1-welfare} imply that, in the same regime
\eqref{eq:regime-sparse}, BGCP achieves
$\Delta_{\mathrm{BGCP}}(T) = 0$ with high probability, and the
$\ell_1$-based pipeline satisfies
\eqref{eq:gap-l1}.
Since $\sqrt{K/T} \gg K \sqrt{(\log m)/T}$ when $K \ll m$ and
$T \asymp K \log m$, we obtain:

\begin{corollary}[Sparse high-dimensional advantage of our methods]
\label{cor:sparse-advantage}
Suppose Assumptions~\ref{assump:bgcp}, \ref{assump:l1},
\ref{assump:design}, and~\ref{assump:indep-activation} hold, and that
\eqref{eq:regime-sparse} is satisfied for sufficiently large constants
$C_1,C_2$.
Then for $T$ in the range $C_1 K \log m \le T \le C_2 m$,
\begin{enumerate}
  \item[(i)] BGCP attains zero welfare gap with high probability:
    $\Delta_{\mathrm{BGCP}}(T) = 0$.
  \item[(ii)] The $\ell_1$ pipeline satisfies
    $\Delta_{\ell_1}(T)
     \le C\,(\sigma K/\kappa_\Sigma^2)\sqrt{(\log m)/T}$.
  \item[(iii)] The EPC baseline satisfies
    $\Delta_{\mathrm{EPC}}(T)
      \ge c_2 \sigma \sqrt{K/T}$
    with high probability.
\end{enumerate}
In particular, for fixed $(K,m)$ and $T$ in this regime, both BGCP and
the $\ell_1$ pipeline have strictly smaller welfare gap than EPC up to
constants.
Moreover, the DLS baseline is not applicable when $T < m$, as
$X^\top X$ is rank-deficient and \eqref{eq:dls-estimator} is
ill-posed.
\end{corollary}

Corollary~\ref{cor:sparse-advantage} formalises the intuitive claim
that when the value function is genuinely sparse ($K \ll m$) and the
number of episodes is in the high-dimensional regime
($T \ll m$), exploiting sparsity is essential: EPC cannot match the
welfare guarantees of BGCP or the $\ell_1$ pipeline, and dense methods
like DLS are not even defined.

\subsection{A Dense Regime Where Classical Methods Compete}
\label{subsec:regime-dense}

We now discuss a complementary regime where classical dense
probabilistic CSG schemes can be competitive or preferable.

Consider a sequence of problems with $m \to \infty$ and
$K = K(m)$ such that $K/m \to \rho \in (0,1]$.
Suppose that $T$ grows faster than $m$, say
\begin{equation}
  T \;\ge\; C_3 m \log m
  \label{eq:regime-dense}
\end{equation}
for some constant $C_3>0$, and that Assumption~\ref{assump:design}
holds with a well-conditioned covariance $\Sigma$.
In this regime, the value function is effectively dense and the
classical OLS (or dense Bayesian) estimator is known to achieve the
parametric rate
\begin{equation}
  \|\hat{\theta}^{\mathrm{DLS}} - \theta^\star\|_2
  \;=\;
  O_\mathbb{P}\!\left(\sigma\sqrt{\frac{m}{T}}\right)
  \label{eq:dls-l2-rate}
\end{equation}
under mild conditions on $X$
\cite{HastieTibshiraniFriedman2009Elements,Bishop2006Pattern}.

By an argument analogous to the proof of
Theorem~\ref{thm:l1-welfare} (using Lemma~\ref{lem:welfare-lipschitz}
and the inequality $\|u\|_1 \le \sqrt{m}\|u\|_2$), we obtain the
following result.

\begin{proposition}[DLS welfare gap in the dense regime]
\label{prop:dls-gap}
Suppose Assumption~\ref{assump:design} holds with $K/m \to \rho \in
(0,1]$ and that \eqref{eq:regime-dense} is satisfied.
Assume further that $X^\top X$ is invertible with high probability.
Then there exist constants $C_4,C_5,C_6>0$ such that, with probability
at least $1 - C_5 m^{-C_6}$,
\begin{equation}
  \Delta_{\mathrm{DLS}}(T)
  \;\le\;
  C_4\,\sigma \sqrt{\frac{m}{T}}.
  \label{eq:dls-gap}
\end{equation}
\end{proposition}

In the dense regime \eqref{eq:regime-dense}, the sparsity assumptions
underpinning BGCP and the $\ell_1$ pipeline become unrealistic:
the mutual coherence condition required by BGCP is difficult to
satisfy when many coordinates of $\theta^\star$ are non-zero, and the
$\ell_1$ estimator incurs a sparsity-induced bias that is not present
in OLS.
In particular, the rate~\eqref{eq:gap-l1} scales as
$K \sqrt{(\log m)/T} \approx m \sqrt{(\log m)/T}$ when $K \approx m$,
which is suboptimal compared to the $m^{1/2} T^{-1/2}$ rate in
\eqref{eq:dls-gap}.
Thus, in the dense regime with abundant data, classical dense
probabilistic CSG schemes based on least squares or dense Bayesian
linear models can be preferable: they match the minimax parametric
rate and do not pay a sparsity penalty.

\subsection{Computational Complexity of Probabilistic CSG Schemes}
\label{subsec:complexity}

We now compare the computational complexity of the proposed algorithms
(BGCP and the $\ell_1$ pipeline) with the EPC and DLS baselines.
Throughout, $T$ denotes the number of episodes, $m$ the number of
candidate coalitions, $K$ the sparsity level $|S^\star|$, and $q$ an
upper bound on the number of active coalitions per episode
($\|X_t\|_0 \le q$).

\paragraph{BGCP.}
A straightforward implementation of BGCP maintains the current support
$S_k$ and residual $r_k = Y - X_{S_k}\hat{\theta}_{S_k}$ at iteration
$k$.
At each step, BGCP evaluates add/remove scores for all $j$ in a
candidate set (e.g.\ all $j \notin S_k$ for additions, all
$j \in S_k$ for removals) using inner products of the form
$X_j^\top r_k$.
Computing all such scores from scratch costs $O(T m)$ operations.
Over $K$ iterations, the total cost is therefore
\begin{equation}
  \mathrm{Time}_{\mathrm{BGCP}}
  \;=\;
  O(T m K)
  \label{eq:time-bgcp}
\end{equation}
up to lower-order terms, assuming $K \le K_{\max}$ is fixed in advance.
If we maintain a Cholesky factor of $X_{S_k}^\top X_{S_k}$, updating
the posterior or model-selection criterion at each step adds at most
$O(K^2)$ per iteration, which is dominated by the $O(T m)$ cost when
$T m \gg K^2$.

\paragraph{$\ell_1$ pipeline.}
For the Lasso estimator \eqref{eq:lasso-def}, standard solvers such as
coordinate descent or proximal gradient methods have per-iteration
cost $O(T m)$, dominated by matrix--vector products with $X$.
Let $I_{\ell_1}$ denote the number of iterations required to reach a
prescribed optimisation accuracy.
Then the computational cost of the $\ell_1$ stage is
\begin{equation}
  \mathrm{Time}_{\ell_1}
  \;=\;
  O(T m I_{\ell_1}).
  \label{eq:time-l1}
\end{equation}
In high-dimensional settings, $I_{\ell_1}$ typically grows at most
polylogarithmically in $1/\varepsilon$ for an $\varepsilon$-accurate
solution.
After computing $\hat{\theta}^{\ell_1}$, one may prune coordinates
$|\hat{\theta}^{\ell_1}_j|$ below a threshold to reduce the effective
coalition set before running CSG.

\paragraph{EPC.}
For EPC, one passes through the data once to accumulate, for each
$j$, the sufficient statistics
$\sum_{t \in \mathcal{T}_j} Y_t/X_{tj}$ and $n_j$.
Since each episode $t$ has at most $q$ non-zero entries in $X_t$, the
total number of operations required to form all plug-in estimates
\eqref{eq:epc-estimator} is
\begin{equation}
  \mathrm{Time}_{\mathrm{EPC}}
  \;=\;
  O(T q)
  \;\le\;
  O(T m).
  \label{eq:time-epc}
\end{equation}
Thus EPC is computationally cheap on the estimation side, but, as
shown in Proposition~\ref{prop:epc-gap}, it has inferior welfare
guarantees in the sparse high-dimensional regime.

\paragraph{DLS.}
For DLS, solving the normal equations
$X^\top X \hat{\theta}^{\mathrm{DLS}} = X^\top Y$ with a direct method
(e.g.\ Cholesky factorisation) requires forming $X^\top X$ at cost
$O(T m^2)$ and factorising it at cost $O(m^3)$, yielding
\begin{equation}
  \mathrm{Time}_{\mathrm{DLS}}
  \;=\;
  O(T m^2 + m^3).
  \label{eq:time-dls}
\end{equation}
Iterative methods such as conjugate gradients can reduce the factor
from $m^3$ to $O(T m I_{\mathrm{DLS}})$, where $I_{\mathrm{DLS}}$ is
the iteration count, but the dependence on $m$ remains quadratic in
the worst case.

\paragraph{Deterministic CSG stage.}
All schemes considered (BGCP, $\ell_1$, EPC, DLS) ultimately invoke a
deterministic CSG solver on some coalition set:
\begin{itemize}
  \item BGCP: on $\widehat{\mathcal{C}}^{\mathrm{BGCP}}
               = \{C_j : j \in S^{\mathrm{BGCP}}\}$, of size
        $|S^{\mathrm{BGCP}}| \approx K$.
  \item $\ell_1$: potentially on a pruned set
        $\widehat{\mathcal{C}}^{\ell_1}$ of size $\tilde{K}$,
        obtained by thresholding $\hat{\theta}^{\ell_1}$.
  \item EPC, DLS: on the full coalition set $\mathcal{C}$ of size $m$
        (unless additional pruning is applied).
\end{itemize}
Exact CSG algorithms based on dynamic programming or branch-and-bound
\cite{Sandholm1999Coalition,Rahwan2015CSGSurvey} have worst-case
complexity exponential in the number of agents $|N|$; our framework
does not change this worst-case dependence.
However, by reducing the number of \emph{relevant} coalitions from $m$
to approximately $K$ (or $\tilde{K}$), BGCP and the $\ell_1$ pipeline
can substantially reduce the state space and pruning effort of the CSG
solver in practice.

\paragraph{Summary of complexity trade-offs.}
Ignoring the CSG stage, EPC has the lowest estimation cost
($O(T q)$), followed by BGCP and the $\ell_1$ pipeline (both
$O(T m)$ up to iteration factors) and DLS ($O(T m^2 + m^3)$).
When the combinatorial CSG cost is taken into account, the picture
changes: methods that aggressively reduce the coalition set (BGCP and,
to a lesser extent, the $\ell_1$ pipeline) can offset their higher
estimation cost by shrinking the exponential component of the CSG
solver.
EPC and DLS are cheap to fit but pass a large coalition universe to
the CSG stage, potentially increasing total runtime.

\subsection{Summary: When to Use Which Probabilistic CSG Scheme?}
\label{subsec:comp-summary}

The theoretical and computational comparison in this section can be
summarised as follows.

\begin{itemize}
  \item In the \emph{sparse high-dimensional regime}
        $K \ll m$ and $T \asymp K \log m$, our probabilistic CSG
        pipeline is theoretically superior.
        BGCP achieves exact welfare optimality with high probability,
        and the $\ell_1$ pipeline enjoys a non-trivial welfare gap
        bound of order
        $\sigma K \sqrt{(\log m)/T} / \kappa_\Sigma^2$.
        In the same regime, EPC exhibits a welfare gap of order
        $\sigma \sqrt{K/T}$ and cannot exploit sparsity, while DLS is
        not applicable as $T < m$.
        Computationally, BGCP and the $\ell_1$ pipeline have
        estimation cost $O(T m)$ but can reduce the coalition set from
        $m$ to approximately $K$ before invoking the CSG solver.

  \item In the \emph{dense, data-rich regime} where $K$ grows
        proportionally to $m$ and $T \gg m$, dense probabilistic CSG
        schemes such as DLS (or dense Bayesian linear regression
        combined with CSG) achieve the parametric rate
        $\sigma \sqrt{m/T}$ and do not pay a sparsity penalty.
        In this regime, the structural advantages of BGCP and the
        $\ell_1$ pipeline largely disappear, and the latter may suffer
        from unnecessary sparsity bias and conservative tuning.
        Since $T$ is large, the quadratic cost in $m$ of DLS can be
        amortised, and the combinatorial CSG stage sees a coalition
        universe that is, in any case, fully used.

  \item The choice between BGCP and the $\ell_1$ pipeline itself
        reflects a trade-off between combinatorial and convex
        optimisation.
        BGCP yields exact support and welfare recovery under strong
        coherence and minimum signal conditions but relies on a greedy
        non-convex search with cost $O(T m K)$.
        The $\ell_1$ pipeline is convex and robust to modest
        violations of the sparsity assumptions but attains only
        approximate welfare optimality, with estimation cost
        $O(T m I_{\ell_1})$ and a welfare gap that decays at rate
        $K \sqrt{(\log m)/T}$.

  \item EPC and DLS occupy the opposite corner: they are simple and
        fast to fit (EPC) or statistically optimal in dense regimes
        (DLS) but either ignore sparsity or require $T \gtrsim m$.
        They pass a large coalition set to the CSG solver and offer no
        guarantees in the sparse high-dimensional regime.
\end{itemize}

Overall, the theory developed in this paper provides a principled
answer to the question ``why probabilistic CSG via sparse episodic
models?'': when only a few coalitions are truly profitable and data is
limited, exploiting sparsity yields provably better welfare guarantees
and more favourable end-to-end complexity than natural episodic
plug-in or dense least-squares alternatives, while classical dense
schemes remain appropriate when values are genuinely dense and
episodes are plentiful.

\bibliographystyle{IEEEtran}
\bibliography{refs}

\end{document}